\documentclass{amsart}%
\usepackage{amsfonts}
\usepackage{amsmath}
\usepackage{amssymb}
\usepackage{color}
\usepackage[polish,english]{babel}
\usepackage[utf8]{inputenc}
\usepackage[T1]{fontenc}
\usepackage{graphicx}%
\providecommand{\U}[1]{\protect\rule{.1in}{.1in}}
\theoremstyle{plain}
\newtheorem{theorem}{Theorem}[section]
\newtheorem{lemma}[theorem]{Lemma}
\newtheorem{proposition}[theorem]{Proposition}
\theoremstyle{definition}

\theoremstyle{remark}
\newtheorem{remark}[theorem]{Remark}
\numberwithin{equation}{section}

\begin{document}
\title{ THE INTEGRABLE HEAVENLY TYPE EQUATIONS AND THEIR\\LIE-ALGEBRAIC STRUCTURE}
\author{Oksana E. Hentosh}
\address{The Institute for Applied Problems of Mechanics and Mathematics at the NAS,
Lviv, 79060 Ukraine}
\email{ohen@ukr.net}
\author{Yarema A. Prykarpatsky}
\address{the Department of Applied Mathematics at the University of Agriculture in
Krakow, 30059, Poland}
\email{yarpry@gmail.com}
\author{Denis Blackmore}
\address{Department of Mathematical Sciences at NJIT, University Heights, Newark, NJ
07102 USA}
\email{denblac@gmail.com }
\author{Anatolij K. Prykarpatski }
\address{The Department of Applied Mathematics at AGH University of Science and
Technology, Krakow 30059, Poland}
\email{pryk.anat@cybergal.com}
\maketitle

\begin{abstract}
The work is devoted to old and recent investigations of the classical M. A.
Buhl problem of describing compatible linear vector field equations, its
general M.G. Pfeiffer and modern Lax-Sato type special solutions. Eespecially
we analyze the the related Lie-algebraic structures and integrability
properties of a very interesting class of nonlinear dynamical systems called
the dispersionless heavenly type equations, which were initiated by
Pleba\'{n}ski and later analyzed in a series of articles. The AKS-algebraic
and related $\mathcal{R}$-structure schemes are used to study the orbits of
the corresponding co-adjoint actions, which are intimately related to the
classical Lie--Poisson structures on them. It is demonstrated that their
compatibility condition coincides with the corresponding heavenly type
equations under consideration. It is shown that all these equations originate
in this way and can be represented as a Lax compatibility condition for
specially constructed loop vector fields on the torus. The infinite hierarchy
of conservations laws related to the heavenly equations is described, and its
analytical structure connected with the Casimir invariants, is mentioned. In
addition, typical examples of such equations, demonstrating in detail their
integrability via the scheme devised herein, are presented. The relationship
of the very interesting Lagrange--d'Alembert type mechanical interpretation of
the devised integrability scheme with the Lax--Sato equations is also discussed.

\end{abstract}

\section{\label{Sec_1}Introduction\ }

In the classical works \cite{Buhl-1,Buhl-2,Buhl-3} still in 1928 the French
mathematician M.A. Buhl posed the problem of classifying all infinitesimal
symmetries of a given linear vector field equation
\begin{equation}
A\psi=0,\text{ \ \ } \label{V0aa}%
\end{equation}
where $\ \ $\ function $\psi$ $\in C^{2}(\mathbb{R}^{n};\mathbb{R}),\ $and%

\begin{equation}
A:=%
{\displaystyle\sum\limits_{j=\overline{1,n}}}
a_{j}(x)\frac{\partial}{\partial x_{j}}\ \label{V0a}%
\end{equation}
is a vector field operator on $\mathbb{R}^{n}$ with coefficients $a_{j}\in
C^{1}(\mathbb{R}^{n};\mathbb{R}),j=\overline{1,n}.$ It is easy to show that
the problem under regard is reduced \cite{Pfei-4} to describing all possible
vector fields
\begin{equation}
A^{(k)}:=%
{\displaystyle\sum\limits_{j=\overline{1,n}}}
a_{j}^{(k)}(x)\frac{\partial\ }{\partial x_{j}} \label{V0b}%
\end{equation}
with coefficients $a_{j}^{(k)}\in C^{1}(\mathbb{R}^{n};\mathbb{R}%
),j,k=\overline{1,n},\ $ satisfying the Lax type commutator condition%
\begin{equation}
\lbrack A,A^{(k)}]=0 \label{V0c}%
\end{equation}
for all $x\in\mathbb{R}^{n}$ and $k=\overline{1,n}.$ \ The M.A. Buhl problem
above was completely solved in 1931 by the Ukrainian mathematician G. Pfeiffer
in the works \cite{Pfei-1,Pfei-2,Pfei-3,Pfei-4,Pfei-5,Pfei-6}, where he has
constructed explicitly the searched set of independent vector fields
\ (\ref{V0b}), \ having made use effectively of the full set of invariants for
the vector field \ (\ref{V0a}) and the related solution set structure of the
Jacobi-Mayer system of equations, naturally following from \ (\ref{V0c}).
\ Some results, yet not complete, were also obtained by C. Popovici \ in
\cite{Popo}.

Some years ago the M.A. Buhl type equivalent problem was\ independently
reanalyzed once more by Japanese mathematicians K. Takasaki and T. Takebe
\cite{TaTa-1,TaTa-2} and later by L.V. Bogdanov, V.S. Dryuma and S.V. Manakov
\cite{BoDrMa} \ for a very special case when the vector field operator
\ (\ref{V0a}) depends analytically on a "spectral" parameter $\lambda
\in\mathbb{C}:$
\begin{equation}
A:=\frac{\partial}{\partial t}+%
{\displaystyle\sum\limits_{j=\overline{1,n}}}
a_{j}(t,x;\lambda)\frac{\partial}{\partial x_{j}}+a_{0}(t,x;\lambda
)\frac{\partial}{\partial\lambda}. \label{V0d}%
\end{equation}
Based on the before developed Sato theory \cite{Sato-1,Sato-2}, the authors
mentioned above \ have shown for some special kinds of vector fields
\ (\ref{V0d}) that there exists an infinite hierarchy of the symmetry vector
fields
\begin{equation}
A^{(k)}:=\frac{\partial}{\partial\tau_{k}}+%
{\displaystyle\sum\limits_{j=\overline{1,n}}}
a_{j}^{(k)}(\tau,x;\lambda)\frac{\partial\ }{\partial x_{j}}+a_{0}^{(k)}%
(\tau,x;\lambda)\frac{\partial}{\partial\lambda},\ \label{V0e}%
\end{equation}
where $\tau=(t;\tau_{1},\tau_{2},...)\in\mathbb{R}^{\mathbb{Z}_{+}}%
,k\in\mathbb{Z}_{+},$ satisfying the Lax-Sato type compatible commutator
conditions
\begin{equation}
\lbrack A,A^{(k)}]=0=[A^{(j)},A^{(k)}] \label{V0f}%
\end{equation}
for all $k,j\in\mathbb{Z}_{+}.$ Moreover, in the cases under regard, the
compatibility conditions \ (\ref{V0f}) proved to be equivalent to some very
important for applications heavenly type dispersionless equations in partial derivatives.

In the present work we investigate the Lax--Sato compatible systems, the
related Lie-algebraic structures and complete integrability properties of an
interesting class of nonlinear dynamical systems called the heavenly type
equations, which were introduced by Pleba\'{n}ski \cite{Pleb} and analyzed in
a series of articles
\cite{MaSa,BoDrMa,Ovsi-1,Ovsi-2,Pavl,Schi,Schi-1,TaTa-1,TaTa-2}. In our work,
having employed the AKS-algebraic and related $\mathcal{R}$-structure schemes
\cite{Blas,BlPrSa,BlSz,FaTa,ReSe,PrSa}, applied to the \ holomorphic loop Lie
algebra \ $\mathcal{\tilde{G}}:=\widetilde{diff}(\mathbb{T}^{n})$ of vector
fields on torus $\mathbb{T}^{n},n\in\mathbb{Z}_{+},$ the orbits of the
corresponding coadjoint actions on $\mathcal{\tilde{G}}^{\ast}$, closely
related to the classical Lie--Poisson type structures, were reanalyzed and
studied in detail. By constructing two commuting flows on the coadjoint space
$\mathcal{\tilde{G}}^{\ast},$ generated by a chosen root element $\tilde{l}%
\in\mathcal{\tilde{G}}^{\ast}$ and some Casimir invariants, we have
successively demonstrated that their compatibility condition coincides exactly
with the corresponding heavenly equations under consideration.

As a by-product of the construction, devised in the work \cite{PrPr} , we
prove that all the heavenly equations have a similar origin and can be
represented as a Lax compatibility condition for special loop vector fields on
the torus $\mathbb{T}^{n}.$ We analyze the structure of the infinite hierarchy
of conservations laws, related to the heavenly equations, and demonstrate
their analytical structure connected with the Casimir invariants is generated
by the Lie--Poisson structure on $\mathcal{\tilde{G}}^{\ast}.$ Moreover, we
have extended the initial Lie-algebraic structure for the case when the the
basic Lie algebra $\mathcal{\tilde{G}=}$ $\widetilde{diff}(\mathbb{T}^{n})$ is
replaced by the adjacent holomorphic Lie algebra $\mathcal{\bar{G}%
}:=diff_{hol}(\ \mathbb{C\times T}^{n})\subset diff(\ \mathbb{C\times T}%
^{n})\ $of vector fields on $\ \mathbb{C\times T}^{n}$. Typical examples are
presented for all cases of the heavenly equations and it is shown in detail
and their integrability is demonstrated using the scheme devised here. This
scheme makes it possible to construct a very natural derivation of well known
Lax--Sato representation for an infinite hierarchy\ of heavenly equations,
related to the canonical Lie--Poisson structure on the adjoint space
$\mathcal{\bar{G}}^{\ast}.$ We also briefly discuss the Lagrangian
representation of these equations following from their Hamiltonicity with
respect to both intimately related commuting evolutionary flows, and the
related bi-Hamiltonian structure as well as the B\"{a}cklund transformations.
As a matter of fact, there are only a few examples of multi-dimensional
integrable systems for which such a detailed description of their mathematical
structure has been given. As was aptly mentioned in \cite{ShMaYa}, the
heavenly equations comprise an important class of such integrable \ systems.
This is due in part to the fact that some of them are obtained by a reduction
of the Einstein equations with Euclidean (and neutral) signature for (anti-)
self-dual gravity, which includes the theory of gravitational
instantons.\ This and other cases of important applications of
multi-dimensional integrable equations strongly motivated us to\ study this
class of equations and the related mathematical structures. As a very
interesting aspect of our approach to describing integrability of the heavenly
dynamical systems, there is a very interesting Lagrange--d'Alembert type
mechanical interpretation. We need to underline here that the main motivating
idea behind this work was based both on the paper by Kulish \cite{Kuli},
devoted to studying the super-conformal Korteweg--de-Vries equation as an
integrable Hamiltonian flow on the adjoint space to the holomorphic loop Lie
superalgebra of super-conformal vector fields on the circle, and on the
insightful investigation by Mikhalev \cite{Mikh}, which studied Hamiltonian
structures on the adjoint space to the holomorphic loop Lie algebra of smooth
vector fields on the circle. We were also impressed by deep technical results
\cite{TaTa-1,TaTa-2} of Takasaki and Takebe, who fully realized the vector
field scheme of the Lax--Sato theory. Additionally, we were strongly
influenced both by the works of Pavlov, Bogdanov, Dryuma, Konopelchenko and
Manakov \cite{BoPa,BoDrMa,BoKo,Kono}, as well as by the work of Ferapontov and
Moss \cite{FeMo}, in which they devised new effective differential-geometric
and analytical methods for studying an integrable degenerate multi-dimensional
dispersionless heavenly type hierarchy of equations, the mathematical
importance of which is still far from being properly appreciated. Concerning
other Lie-algebraic approaches to constructing integrable heavenly equations,
we mention work by Szablikowski and Sergyeyev \cite{Szab,SeSz}, \ Ovsienko
\cite{Ovsi-1,Ovsi-2} and by Kruglikov and Morozov \cite{KrMo}. \ 

\section{The Lax--Sato type compatible systems of linear vector field
equations}

\subsection{A vector field on the torus and its invariants}

Consider a simple vector field $X:\mathbb{R\times}\mathbb{T}^{n}\rightarrow
T(\mathbb{R\times T}^{n})$ on the $(n+1)$-dimensional toroidal cylinder
$\mathbb{R\times T}^{n}$ \ for arbitrary $n\in\mathbb{Z}_{+},$ which we will
write in the slightly special form
\begin{equation}
A=\frac{\partial}{\partial t}+<a(t,x),\frac{\partial}{\partial x}>, \label{V1}%
\end{equation}
where $(t,x)\in\mathbb{R\times T}^{n},a(t,x)\in\mathbb{E}^{n},\frac{\partial
}{\partial x}:=(\frac{\partial}{\partial x_{1}},\frac{\partial}{\partial
x_{2}},...,\frac{\partial}{\partial x_{n}})^{\intercal}$ and $\ <\cdot,\cdot>$
\ \ is the standard scalar product on the Euclidean space $\mathbb{E}^{n}.$
With the vector field \ (\ref{V1}), one can associate the linear equation
\begin{equation}
A\psi=0 \label{V2}%
\end{equation}
for some function $\psi\in C^{2}(\mathbb{R\times T}^{n};\mathbb{R}),$ which we
will call an \textquotedblleft invariant\textquotedblright\ of the vector field.

Next, we study the existence and number of such functionally-independent
invariants to the equation (\ref{V2}). For this let us pose the following
Cauchy problem for equation (\ref{V2}): Find a function $\psi\in
C^{2}(\mathbb{R\times T}^{n};\mathbb{R}),$ which at point $\ t^{(0)}%
\in\mathbb{R\ }$ \ satisfies the condition $\psi(t,x)|_{t=t^{(0)}}=\psi
^{(0)}(x),$ $x$ $\in\mathbb{T}^{n},$ for a given function $\psi^{(0)}\in
C^{2}(\mathbb{T}^{n};\mathbb{R}).$ For the equation (\ref{V2}) there is a
naturally related parametric vector field on the torus $\mathbb{T}^{n}\ $\ in
the form of the ordinary vector differential equation%
\begin{equation}
\frac{dx}{dt}=a(t,x), \label{V3}%
\end{equation}
to which there corresponds the following Cauchy problem: find a function
$x:\mathbb{R}\rightarrow\mathbb{T}^{n}$ satisfying
\begin{equation}
x(t)|_{t=t^{(0)}}=z\ \label{V4}%
\end{equation}
for an arbitrary \ constant vector $z\in\mathbb{T}^{n}.$ Assuming that the
vector-function $a\in C^{1}(\mathbb{R\times T}^{n};\mathbb{E}^{n}),$ it
follows from the classical Cauchy theorem \cite{Cart} on the existence and
unicity of the solution to \ (\ref{V3}) and \ (\ref{V4}) that we can obtain a
unique solution to the vector equation \ (\ref{V3}) as some function
$\ \Phi\in C^{1}(\mathbb{R}\times\mathbb{T}^{n};\mathbb{T}^{n}),\ x=\Phi
(t,z),$ such that the matrix $\ \partial\Phi(t,z)/\partial z$ is nondegenerate
for all $t\in\mathbb{R}$ sufficiently close to $t^{(0)}\in\mathbb{R}.$ Hence,
the Implicit Function Theorem \cite{Cart,CoLe} implies that there exists a
mapping $\Psi:\mathbb{R}\times\mathbb{T}^{n}\rightarrow\mathbb{T}^{n},$ such
that
\begin{equation}
\Psi(t,x)=z \label{V5}%
\end{equation}
for every $z\in\mathbb{T}^{n}$ and all $t\in\mathbb{R}$ sufficiently enough to
$t^{(0)}\in\mathbb{R}.$ Supposing now that the functional vector
$\Psi(t,x)=(\psi^{(1)}(t,x),\psi^{(2)}(t,x),...,\psi^{(n)}(t,x))^{\intercal},$
$(t,x)\in\mathbb{R}\times\mathbb{T}^{n},$\ is constructed, from the
arbitrariness of the parameter $z\in\mathbb{T}^{n}$ one can deduce that all
functions $\psi^{(j)}:\mathbb{R}\times\mathbb{T}^{n}\rightarrow\mathbb{T}%
^{1},$ $j=\overline{1,n},$ are functionally independent invariants of the
vector field equation \ (\ref{V2}), that is $A\psi^{(j)}=0,j=\overline{1,n}.$
Thus, the vector field equation (\ref{V2}) has exactly $\ n\in\mathbb{Z}_{+}$
functionally independent invariants, which make it possible, in particular, to
solve the Cauchy problem posed above. \ Namely, let \ a mapping $\alpha
:\mathbb{T}^{n}\rightarrow\mathbb{R}$ be chosen such that $\alpha
(\Psi(t,x))|_{t=t^{(0)}}=\psi^{(0)}(x)\ $ for all $\ x$ $\in\mathbb{T}^{n}$
and a fixed $t^{(0)}\in\mathbb{R}.$ Inasmuch as the superposition of functions
$\alpha\circ\Psi:\mathbb{R}\times\mathbb{T}^{n}\rightarrow\mathbb{T}^{1}$
\ is, evidently, also an invariant for the equation (\ref{V2}), it provides
the solution to this Cauchy problem, which we can formulate as the following result.

\begin{proposition}
The linear equation (\ref{V2}), generated by the vector field \ (\ref{V3}) on
the torus $\mathbb{T}^{n},$ has exactly $n\in\mathbb{Z}_{+}$ functionally
independent invariants.
\end{proposition}

Consider now a Plucker type \cite{Kono} differential form $\chi^{(n)}%
\in\Lambda^{n}(\mathbb{T}^{n})$ on the torus $\mathbb{T}^{n}$ \ as
\begin{equation}
\chi^{(n)}:=\mathrm{d}\psi^{(1)}\wedge\mathrm{d}\psi^{(2)}\wedge
...\wedge\mathrm{d}\psi^{(n)}, \label{V6}%
\end{equation}
generated by the vector $\Psi:\mathbb{R}^{n}\times\mathbb{T}^{n}%
\rightarrow\mathbb{T}^{n}$ of independent invariants \ \ (\ref{V5}), depending
additionally on $n\in\mathbb{Z}_{+}$ parameters $t\in\mathbb{R}^{n},$ where by
definition, \ \ for $\ $any $\ k=\overline{1,n}$
\begin{equation}
\mathrm{d}\psi^{(k)}:=\sum_{j=\overline{1,n}}\frac{\partial\psi^{(k)}%
}{\partial x_{j}}dx_{j}\ \label{V6a}%
\end{equation}
on the manifold $\mathbb{T}^{n}.$ \ As follows from the Frobenius theorem
\cite{Arno,Cart,Godb,Kono}, the Plucker type differential form (\ref{V6}) is
for all fixed parameters $t\in\mathbb{R}^{n}$ nonzero on the manifold
$\mathbb{T}^{n}$ owing to the functional independence of the invariants
\ (\ref{V5}). It is easy to see that at the fixed parameters $t\in
\mathbb{R}^{n}$ the following \cite{Pfei-6} Jacobi-Mayer type relationship%
\begin{equation}
\left\vert \frac{\partial\Psi}{\partial x}\right\vert ^{-1}\mathrm{d}%
\psi^{(1)}\wedge\mathrm{d}\psi^{(2)}\wedge...\wedge\mathrm{d}\psi^{(n)}%
=dx_{1}\wedge dx_{2}\wedge...\wedge dx_{n}\ \label{V7}%
\end{equation}
holds for $\ k=\overline{1,n}$ \ on the manifold $\mathbb{T}^{n},$ where
$\ \left\vert \frac{\partial\Psi}{\partial x}\right\vert $ is the determinant
of the Jacobi mapping $\frac{\partial\Psi}{\partial x}:T(\mathbb{T}%
^{n})\rightarrow T(\mathbb{T}^{n})$ of the mapping \ (\ref{V5}) subject to the
torus variables $\ x\in\mathbb{T}^{n}.$ On the right-hand side of \ (\ref{V7})
one has the volume measure on the torus $\mathbb{T}^{n},$ which is naturally
dependent on $t\in\mathbb{R}^{n}$ owing to the general vector field
relationships \ (\ref{V3}). \ Taking into account that for all $k=\overline
{1,n}$ \ the full differentials $\ $
\begin{equation}
\ d\psi^{(k)}=\sum_{s=\overline{1,n}}\frac{\partial\psi^{(k)}}{\partial t_{s}%
}dt_{s}+\mathrm{d}\psi^{(k)}\ \ \label{V7a}%
\end{equation}
vanish on $\mathbb{R}^{n}\times\mathbb{T}^{n},$ \ the corresponding
substitution of the reduced differentials $\mathrm{d}\psi^{(k)}\in
C^{2}(\mathbb{R}^{n};\Lambda^{1}(\mathbb{T}^{n})),k=\overline{1,n},$ \ into
\ (\ref{V7}) easily gives rise, in particular, to the following set of the
compatible vector field relationships
\begin{equation}
\ \frac{\partial\Psi}{\partial t_{s}}-\ \sum_{j,k=\overline{1,n}}\left[
\left(  \frac{\partial\Psi}{\partial x}\right)  _{jk}^{-1}\frac{\partial
\psi^{(k)}}{\partial t_{s}}\right]  \frac{\partial\Psi}{\partial x_{j}%
}=0,\ \label{V7b}%
\end{equation}
for all $s=\overline{1,n}.$ \ The latter property, as it was demonstrated by
M.G. Pfeiffer in \cite{Pfei-6}, makes it possible to solve effectively the
M.A. Buhl problem and has interesting applications \cite{BoDrMa,Kono} in the
theory of completely integrable dynamical systems of heavenly type, which are
considered in the next section.

\subsection{Vector field hierarchies on the torus with \textquotedblleft
spectral\textquotedblright\ parameter and the Lax--Sato integrable heavenly
dynamical systems}

Consider some naturally ordered infinite set \ of parametric vector fields
\ (\ref{V1}) on the torus $\mathbb{T}^{n}$ in the form
\begin{equation}
A^{(k)}=\frac{\partial}{\partial t_{k}}+<a^{(k)}\ (t,x;\lambda),\frac
{\partial}{\partial x}>+\text{ }a_{0}^{(k)}(t,x;\lambda)\frac{\partial
}{\partial\lambda}:=\frac{\partial}{\partial t_{k}}+\mathrm{A}^{(k)},
\label{V8}%
\end{equation}
where $t_{k}\in\mathbb{R},k\in\mathbb{Z}_{+},(t,x;\lambda)\in(\mathbb{R}%
^{\mathbb{Z}_{+}}\times\mathbb{T}^{n})\times\mathbb{C}\ \ $are the evolution
parameters, and the $\ $dependence of smooth vectors $(a_{0}^{(k)}%
,a^{(k)})^{\intercal}\in\mathbb{E}\times\mathbb{E}^{n\ },k\in\mathbb{Z}_{+},$
on the \textit{\textquotedblleft spectral\textquotedblright} parameter
$\lambda\in\mathbb{C}$ is assumed to be holomorphic. Suppose now that the
infinite hierarchy of the linear equations
\begin{equation}
A^{(k)}\psi=0 \label{V9}%
\end{equation}
for $k\in\mathbb{Z}_{+}$ possesses exactly $n+1\in\mathbb{Z}_{+}$ common
functionally independent invariants $\psi^{(j)}(\lambda)\in$ $C^{2}%
(\mathbb{R}^{\mathbb{Z}_{+}}\times\mathbb{T}^{n};\mathbb{C}),$ $j=\overline
{0,n},$ for $\lambda\in\mathbb{C}.$ Then, owing to the existence theory
\cite{Cart,CoLe,Golu} for ordinary differential equations depending on the
\textit{\textquotedblleft spectral\textquotedblright} parameter $\lambda
\in\mathbb{C},$ the solutions may be assumed to be such that allow analytical
continuation in the parameter $\lambda\in\mathbb{C}$ both inside
$\mathbb{S}_{+}^{1}\subset\mathbb{C}$ \ of some circle $\mathbb{S}^{1}%
\subset\mathbb{C}$ and subject to the parameter $\lambda^{-1}\in
\mathbb{C},|\lambda|\rightarrow\infty,$ outside $\mathbb{S}_{-}^{1}%
\subset\mathbb{C}$ of this circle $\mathbb{S}^{1}\subset\mathbb{C}.$ This
means that as $|\lambda|\rightarrow\infty$ we have the following
expansions:$\ $%
\begin{align}
\psi^{(0)}(\lambda)  &  \sim\lambda+\ \sum_{k=0}^{\infty}\psi_{k}%
^{(0)}(t,x)\lambda^{-k},\nonumber\\
& \nonumber\\
\psi^{(1)}(\lambda)  &  \sim\sum_{k=0}^{\infty}\tau_{k}^{(1)}(t,x)\psi
_{0}(\lambda)^{k}+\sum_{k=1}^{\infty}\psi_{k}^{(1)}(t,x)\psi_{0}(\lambda
)^{-k},\label{V10}\\
& \nonumber\\
\psi^{(2)}(\lambda)  &  \sim\sum_{k=0}^{\infty}\tau_{k}^{(2)}(t,x)\psi
_{0}(\lambda)^{k}+\sum_{k=1}^{\infty}\psi_{k}^{(2)}(t,x)\psi_{0}(\lambda
)^{-k},\nonumber\\
&  ...\nonumber\\
\psi^{(n)}(\lambda)  &  \sim\sum_{k=0}^{\infty}\tau_{k}^{(n)}(t,x)\psi
_{0}(\lambda)^{k}+\sum_{k=1}^{\infty}\psi_{k}^{(n)}(t,x)\psi_{0}(\lambda
)^{-k},\nonumber
\end{align}
where we took into account that $\psi^{_{(0)}}(\lambda)\in C^{2}%
(\mathbb{R}^{\mathbb{Z}_{+}}\times\mathbb{T}^{n};\mathbb{C}),$ $\lambda
\in\mathbb{S}_{-}^{1}\subset\mathbb{C},$ is the basic invariant solution to
the equations (\ref{V9}), \ functions $\ $ $\tau_{l}^{(s)}\in C^{2}%
(\mathbb{R}^{\mathbb{Z}_{+}}\times\mathbb{T}^{n};\mathbb{C})\ $for all
$s=\overline{1,n\text{ }}$ and suitable $l=\overline{1,\infty},$ and $\psi
_{k}^{(j)}\in C^{2}(\mathbb{R}^{\mathbb{Z}_{+}}\times\mathbb{T}^{n}%
;\mathbb{C})$ for all $k=\overline{1,\infty},j=\overline{0,n}.$ Write down now
the condition (\ref{V7}) on the manifold $\ \mathbb{C\times T}^{n}$ in the
equivalent form
\begin{equation}
|\frac{\partial\mathrm{\Psi}}{\partial\mathrm{x}}|^{-1}\mathrm{d}\psi
^{(0)}\wedge\mathrm{d}\psi^{(1)}\wedge\mathrm{d}\psi^{(2)}\wedge
...\wedge\mathrm{d}\psi^{(n)}=d\lambda\wedge dx_{1}\wedge dx_{2}%
\wedge...\wedge dx_{n}\ \label{V11}%
\end{equation}
at fixed parameters \ $t\in\mathbb{R}^{\mathbb{Z}_{+}},$ where $\mathrm{x}%
:=(\lambda,x)\in\mathbb{C}\times\mathbb{T}^{n},$ $|\frac{\partial\mathrm{\Psi
}}{\partial\mathrm{x}}|$ is the Jacobi determinant of the mapping
$\mathrm{\Psi}:=\ (\psi^{(0)},\psi^{(1)},\psi^{(2)},...,\psi^{(n)}%
)^{\intercal}$ $\in C^{2}(\mathbb{C}\times(\mathbb{R}^{\mathbb{Z}_{+}}%
\times\mathbb{T}^{n});\mathbb{C}^{n+1})$ subject to the variables
$\mathrm{x}\in\mathbb{C}\times\mathbb{T}^{n}.$ \ Inasmuch as this mapping
subject to the parameter $\lambda\in\mathbb{C}$ has analytical continuation
\cite{CoLe,Golu} inside $\ \mathbb{S}_{+}^{1}\subset\mathbb{C}$ of the circle
$\mathbb{S}^{1}\subset\mathbb{C}$ and subject to the parameter $\lambda
^{-1}\in\mathbb{C}\ $\ as $|\lambda|\rightarrow\infty$ \ outside
$\mathbb{S}_{-}^{1}\subset\mathbb{C}$ of this circle $\mathbb{S}^{1}%
\subset\mathbb{C},$ one can easily obtain from the holomorphic structure of
the vector fields \ (\ref{V8}) subject to the complex variable $\lambda
\in\mathbb{C}$ and the relationship \ (\ref{V11}), reduced on the manifold
$\mathbb{S}_{-}^{1}\mathbb{\times T}^{n},$ the following analytical criterion:%
\begin{equation}
\left(  |\frac{\partial\mathrm{\Psi}}{\partial\mathrm{x}}|^{-1}\mathrm{d}%
\psi^{(0)}\wedge\mathrm{d}\psi^{(1)}\wedge\mathrm{d}\psi^{(2)}\wedge
...\wedge\mathrm{d}\psi^{(n)}\right)  _{-}=0,\ \label{V12}%
\end{equation}
where $(...)_{-}$ means the asymptotic part of the expression in brackets,
depending on the negative degree parameter $\lambda^{-1}\in\mathbb{S}_{-}^{1}$
as $|\lambda|\rightarrow\infty.$ Now, \ since the full differentials
$d\psi^{(j)}\in\Lambda^{1}(\mathbb{C}\times(\mathbb{R}^{n\times\mathbb{Z}_{+}%
}\times\mathbb{T}^{n}),$ $j=\overline{0,n},$ \ vanish on \ $\mathbb{C}%
\times(\mathbb{R}^{n\times\mathbb{Z}_{+}}\times\mathbb{T}^{n},$ we have%
\begin{equation}
\mathrm{d}\psi^{(j)}=-\sum_{k=0}^{\infty}\frac{\partial\psi^{(j)}}%
{\partial\tau_{k}^{(j)}}d\tau_{k}^{(j)}, \label{V13}%
\end{equation}
whose substitution into (\ref{V12}) yields
\begin{equation}
\frac{\partial\mathrm{\Psi}}{\partial\tau_{k}^{(j)}}=\left[  \left(
\frac{\partial\mathrm{\Psi}}{\partial\mathrm{x}}\right)  _{0j}^{-1}%
\psi^{_{(0)}}(\lambda)^{k}\right]  _{+}\ \frac{\partial\mathrm{\Psi}}%
{\partial\lambda}+\sum_{s=1}^{n}\ \left[  \left(  \frac{\partial\mathrm{\Psi}%
}{\partial\mathrm{x}}\right)  _{sj}^{-1}\psi^{(0)}(\lambda)^{k}\right]
_{+}\ \frac{\partial\mathrm{\Psi}}{\partial x_{s}} \label{V14}%
\end{equation}
for all $k\in\mathbb{Z}_{+},j=\overline{1,n};$ these expressions comprise an
infinite hierarchy of Lax--Sato compatible \cite{TaTa-1,TaTa-2} linear
equations, where $(...)_{+}$ denotes the asymptotic part of the expression in
brackets, depending on positive powers of the parameter $\lambda\in
\mathbb{C}.$ \ As for the functional parameters $\tau_{k}^{(j)}\in
C^{1}(\mathbb{R}^{\mathbb{Z}_{+}}\times\mathbb{T}^{n};\mathbb{C})$ for all
$k\in\mathbb{Z}_{+},j=\overline{1,n},$ one can prove their functional
independence by taking into account their \textit{a priori} linear dependence
on the corresponding independent evolution parameters $t_{k}\in\mathbb{R},$
$\ k\in\mathbb{Z}_{+}.$ On the other hand, taking into account the explicit
form of the hierarchy of equations \ (\ref{V14}), following \cite{BoDrMa}, it
is not hard to show that the corresponding vector \ fields
\begin{equation}
\mathrm{A}_{k}^{(j)}:=\left[  \left(  \frac{\partial\mathrm{\Psi}}%
{\partial\mathrm{x}}\right)  _{0j}^{-1}\psi^{_{(0)}}(\lambda)^{k}\right]
_{+}\ \frac{\partial}{\partial\lambda}+\sum_{s=1}^{n}\ \left[  \left(
\frac{\partial\mathrm{\Psi}}{\partial\mathrm{x}}\right)  _{sj}^{-1}\psi
^{(0)}(\lambda)^{k}\right]  _{+}\ \frac{\partial}{\partial x_{s}} \label{V15}%
\end{equation}
on the manifold $\mathbb{C}\times\mathbb{T}^{n}$ satisfy for all
$k,m\in\mathbb{Z}_{+},j,l=\overline{1,n},$ the Lax compatibility conditions
\begin{equation}
\frac{\partial\mathrm{A}_{m}^{(l)}}{\partial\tau_{k}^{(j)}}-\frac
{\partial\mathrm{A}_{k}^{(j)}}{\partial\tau_{m}^{(l)}}=[\mathrm{A}_{k}%
^{(j)},\mathrm{A}_{m}^{(l)}], \label{V16}%
\end{equation}
which are equivalent to the independence of the all functional parameters
$\tau_{k}^{(j)}\in C^{1}(\mathbb{R}^{\mathbb{Z}_{+}}\times\mathbb{T}%
^{n};\mathbb{C}),$ $k\in\mathbb{Z}_{+},j=\overline{1,n}.$ \ As a corollary of
the analysis above, one can show that the infinite hierarchy of vector fields
\ (\ref{V8}) is a linear combination of the basic vector fields \ (\ref{V15})
and also satisfies the Lax type compatibility condition \ (\ref{V16}).
Inasmuch the coefficients of vector fields (\ref{V15}) are suitably smooth
functions on the manifold $\mathbb{\ }\ \mathbb{R}^{\mathbb{Z}_{+}}%
\times\mathbb{T}^{n},$ the compatibility conditions (\ref{V16}) yield the
corresponding sets of differential-algebraic relationships on their
coefficients, which have the common infinite set of invariants, thereby
comprising an infinite hierarchy of completely integrable so called
\textit{heavenly} nonlinear dynamical systems on the corresponding
multidimensional functional manifolds. \ That is, all of the above can be
considered as an introduction to a recently devised
\cite{TaTa-1,TaTa-2,BoDrMa,BoPa} constructive algorithm for generating
infinite hierarchies of completely integrable nonlinear dynamical systems of
heavenly type on functional manifolds of arbitrary dimension. It is worthwhile
to stress here that the above constructive algorithm for generating completely
integrable nonlinear multidimensional dynamical systems still does not make it
possible to directly show they are Hamiltonian and construct other related
mathematical structures. This important problem is solved by employing other
mathematical theories; for example, the analytical properties of the related
loop diffeomorphisms groups generated by the hierarchy of vector fields
\ (\ref{V8}).

\begin{remark}
The compatibility condition \ (\ref{V16}) allows an alternative
differential-geometric description based on the Lie-algebraic properties of
the basic vector fields \ (\ref{V15}). Namely, consider the manifold
$\ \mathbb{R}^{n\times\mathbb{Z}_{+}},$ as the base manifold of the vector
bundle $E(\mathbb{R}^{n\times\mathbb{Z}_{+}},G),$ $\ E=\cup_{\tau
\in\ \mathbb{R}^{n\times\mathbb{Z}_{+}}}\{(G^{\ast}\otimes\tau)/\rho\},$
$G^{\ast}:=\{\varphi^{\ast}:\varphi^{\ast}\beta^{(1)}:=\alpha^{(1)}%
\circ\varphi,$ \ $\beta^{(1)}\in\tilde{\Lambda}^{(1)}(\mathbb{C\times
}\mathbb{T}^{n};\mathbb{C)},\varphi\in G\}\ \ $ \ for an equivalence relation
$\rho$ \ and the $($holomorphic in $\lambda\in\mathbb{S}_{+}^{1}\cup
\mathbb{S}_{-}^{1}\subset\mathbb{C)}$ structure group $G=Diff_{hol}%
(\mathbb{C\times}\mathbb{T}^{n}),$ naturally acting on the vector space $E.$
The structure group can be endowed with a connection $\Upsilon$ \ by means of
a mapping $d_{h}:\Gamma(E)\rightarrow\Gamma(T^{\ast}(\mathbb{R}^{n\times
\mathbb{Z}_{+}})\otimes E)$ $\cong\Gamma(Hom(T(\mathbb{R}^{n\times
\mathbb{Z}_{+}});E)),$ where
\begin{equation}
d_{h}\varphi_{\tau}^{\ast}\ :=\sum_{j\in\mathbb{Z}_{+}}d\tau_{j}^{(k)}%
\otimes\frac{\partial\ }{\partial\tau_{j}^{(k)}}\circ\varphi_{\tau}^{\ast
}+\varphi_{\tau}^{\ast}\circ<\alpha^{(1)},\frac{\partial}{\partial\mathrm{x}%
}>, \label{V16a}%
\end{equation}
$\alpha^{(1)}:=\sum_{j\in\mathbb{Z}_{+}}a_{j}^{(k)}d\tau_{j}^{(k)}\in
\Lambda(\mathbb{R}^{n\times\mathbb{Z}_{+}})\otimes\Gamma(E),$ which is defined
for any cotangent diffeomorphism $\ \ \varphi_{\tau}^{\ast}\in E,$ $\tau
\in\mathbb{R}^{n\times\mathbb{Z}_{+}},$ generated by the set \ of parametric
vector fields \ (\ref{V15}), and naturally acting on any mapping $\psi\in
C^{2}(\mathbb{R}^{n\times\mathbb{Z}_{+}}\times(\ \mathbb{C\times T}%
^{n}\mathbb{)};\mathbb{C})$ as $\varphi_{\tau}^{\ast}\circ\psi(\tau
,\mathrm{x}):=\psi(\tau,\varphi_{\tau}(\mathrm{x})),\ (\tau,x)\in
\mathbb{R}^{n\times\mathbb{Z}_{+}}\times\mathbb{T}^{n}.$ It is easy now to see
that the corresponding to \ (\ref{V16a}) zero curvature condition $d_{h}%
^{2}=0$ is equivalent to the set of compatibility equations \ (\ref{V16}).
Moreover, the parallel transport equation
\begin{equation}
d_{h}\varphi_{\tau}^{\ast}\circ\psi\ =0 \label{V16b}%
\end{equation}
coincides exactly with the infinite hierarchy of linear vector field equations
\ (\ref{V14}), where $\psi\in C^{2}(\mathbb{R}^{n\times\mathbb{Z}_{+}}%
\times\mathbb{T}^{n};\mathbb{R})$ is their invariant. Conversely, the Cartan
integrable ideal of differential forms $\ h(\alpha)\in\Lambda(\mathbb{R}%
^{n\times\mathbb{Z}_{+}}\times\mathbb{T}^{n})\otimes\Gamma(T^{\ast}%
(\mathbb{R}^{n\times\mathbb{Z}_{+}})),$ which is equivalent to the zero
curvature condition $d_{h}^{2}=0,$ \ makes it possible to retrieve
\cite{BlPrSa,PrMy} the corresponding connection $\Upsilon$ by constructing a
mapping $d_{h}:\Gamma(E)\rightarrow\Gamma(T^{\ast}(\mathbb{R}^{n\times
\mathbb{Z}_{+}})\otimes E)$ $\cong\Gamma(Hom(T(\mathbb{R}^{n\times
\mathbb{Z}_{+}});E))$ in the form \ (\ref{V16a}). These and other interesting
related aspects of the integrable heavenly dynamical systems shall be
investigated separately elsewhere.
\end{remark}

\subsection{Example: the vector field representation for the Mikhalev--Pavlov
equation}

The Mikhalev--Pavlov equation was first constructed in \cite{Mikh} and has the
form
\begin{equation}
u_{xt}+u_{yy}=u_{y}u_{xx}-u_{x}u_{xy}, \label{V17}%
\end{equation}
where $u\in C^{\infty}(\mathbb{R}^{2}\times\mathbb{T}^{1};\mathbb{R})$ and
$(t,y;x)\in\mathbb{R}^{2}\times\mathbb{T}^{1}.$ Assume now \cite{BoDrMa} that
the following two functions
\begin{equation}
\psi^{(0)}=\lambda,\text{ \ \ }\psi^{(1)}\sim\sum_{k=3}^{\infty}\lambda
^{k}\tau_{k}-\lambda^{2}t+\lambda y+x+\sum_{j=1}^{\infty}\psi_{j}%
^{(1)}(t,y,\tau;x)\ \lambda^{-j}, \label{V19}%
\end{equation}
where $\psi_{1}^{(1)}(t,y,\tau;x)=u,$ $(t,y,\tau;x)\in\mathbb{R}^{2}%
\times\mathbb{R}^{\infty}\times\mathbb{T}^{1},$ are invariants of the set of
vector fields \ (\ref{V9}) for an infinite set of constant parameters
$\tau_{k}\in\mathbb{R},k=\overline{3,\infty},$ as the complex parameter
$\lambda\rightarrow\infty.$ By applying to\ the invariants (\ref{V19}) the
criterion \ (\ref{V12}), \ (\ref{V13}) in the form
\begin{equation}
((\partial\psi^{(1)}/\partial x)^{-1}\mathrm{d}\psi^{(1)})_{-}=0, \label{V20}%
\end{equation}
one can easily obtain the following compatible\ linear vector field
equations\ \
\begin{align}
\frac{\partial\psi}{\partial t}+(\lambda^{2}+\lambda u_{x}-u_{y}%
)\frac{\partial\psi}{\partial x}  &  =0\label{V20a}\\
& \nonumber\\
\frac{\partial\psi}{\partial y}+(\lambda+u_{x})\frac{\partial\psi}{\partial
x}\  &  =0,\nonumber\\
&  ...\nonumber\\
\frac{\partial\psi}{\partial\tau_{k}}+P_{k}(u;\lambda)\frac{\partial\psi
}{\partial x}  &  =0,\nonumber
\end{align}
where $P_{k}(u;\lambda),k=\overline{3,\infty},$ are independent
differential-algebraic polynomials in the variable $u\in C^{\infty}%
(\mathbb{R}^{2}\times\mathbb{R}^{\infty}\times\mathbb{T}^{1})$ and algebraic
polynomials in the spectral parameter $\lambda\in\mathbb{C},$ \ calculated
from the expressions \ (\ref{V14}). Moreover, as one can check, the
compatibility condition \ (\ref{V16}) for the first two vector field equations
of \ (\ref{V20a}) yields exactly the Mikhalev--Pavlov equation \ (\ref{V17}).

\subsection{Example: The Dunajski metric nonlinear equation}

The equations for the Dunajski metric \cite{Duna} are
\begin{align}
u_{x_{1}t}+u_{yx_{2}}+u_{x_{1}x_{1}}u_{x_{2}x_{2}}-u_{x_{1}x_{2}}-v  &
=0,\label{V21}\\
& \nonumber\\
v_{x_{1}t}+v_{x_{2}y}+u_{x_{1}x_{1}}v_{x_{2}x_{2}}-2u_{x_{1}x_{2}}%
v_{x_{1}x_{2}}  &  =0,\nonumber
\end{align}
where $(u,v)\in C^{\infty}(\mathbb{R}^{2}\times\mathbb{T}^{2};\mathbb{R}%
^{2}),$ $(y,t;x_{1},x_{2})\in\mathbb{R}^{2}\times\mathbb{T}^{2}.$ One can
construct now, by definition, the following asymptotic expansions \
\begin{align}
\psi^{(0)}  &  \sim\lambda+\sum_{j=1}^{\infty}\psi_{j}^{(0)}(t,y;x)\lambda
^{-j},\text{ \ \ }\label{V22}\\
& \nonumber\\
\psi^{(1)}  &  \sim\sum_{k=2}^{\infty}(\psi^{(0)})^{k}\tau_{k}^{(1)}%
-\psi^{(0)}y+x_{1}+\sum_{j=1}^{\infty}\psi_{j}^{(1)}(t,y;x)\ (\psi^{(0)}%
)^{-j},\nonumber\\
& \nonumber\\
\psi^{(2)}  &  \sim\sum_{k=2}^{\infty}(\psi^{(0)})^{k}\tau_{k}^{(2)}%
+\psi^{(0)}t+x_{2}+\sum_{j=1}^{\infty}\psi_{j}^{(1)}(t,y;x)\ (\psi^{(0)}%
)^{-j},\nonumber
\end{align}
where $\partial u/\partial x_{1}:=\psi_{1}^{(2)},\partial u/\partial
x_{2}:=\psi_{1}^{(1)},v:=\psi_{1}^{(0)}$ and $\ \tau_{k}^{(s)}\in
\mathbb{R},s=\overline{1,2},k=\overline{2,\infty},$ are constant parameters.
Then the Lax--Sato conditions \ (\ref{V12}),\ (\ref{V13})
\begin{equation}
\left(  \left\vert \frac{\partial(\psi^{(0)},\psi^{(1)},\psi^{(2)})}%
{\partial(\lambda,x_{1},x_{2})}\right\vert ^{-1}\mathrm{d}\psi^{(0)}%
\wedge\mathrm{d}\psi^{(1)}\wedge\mathrm{d}\psi^{(2)}\right)  _{-}%
=0\ \label{V22a}%
\end{equation}
yield a compatible hierarchy of the following linear vector field equations:%
\begin{equation}%
\begin{array}
[c]{c}%
X^{(t_{0})}\psi:=\frac{\partial\psi}{\partial t}+\mathrm{X}^{(t_{0})}%
\psi=0,\text{ \ }\mathrm{X}^{(t_{0})}:=u_{x_{2}x_{2}}\frac{\partial}{\partial
x_{1}}-(\lambda+u_{x_{1}x_{2})}\frac{\partial}{\partial x_{2}}+v_{x_{2}}%
\frac{\partial}{\partial\lambda}=0,\\
\\
X^{(t_{1})}\psi:=\frac{\partial\psi}{\partial y}+\mathrm{X}^{(t_{1})}%
\psi=0,\text{ \ }\mathrm{X}^{(t_{1})}:=(\lambda-u_{x_{1}x_{2}})\frac{\partial
}{\partial x_{1}}+u_{x_{1}x_{1}}\frac{\partial}{\partial x_{2}}-v_{x_{1}}%
\frac{\partial}{\partial\lambda}=0,\\
\\
X^{(t_{k}^{(s)})}\psi:=\frac{\partial\psi}{\partial\tau_{k}^{s}}+<P_{k}%
^{s}(u;\lambda),\frac{\partial\psi}{\partial\mathrm{x}}>=0,
\end{array}
\label{V23}%
\end{equation}
where $P_{k}^{s}(u,v;\lambda)\in\mathbb{E}^{3},s=\overline{1,2},k=\overline
{2,\infty},$ are analytic in $\lambda\in\mathbb{C}$ independent
differential-algebraic vector polynomials \cite{BoPa} in the variables
$(u,v)\in C^{\infty}(\mathbb{R}^{2}\times\mathbb{R}^{\infty}\times
\mathbb{T}^{2};\mathbb{R}^{2})$ and algebraic polynomials in the spectral
parameter $\lambda\in\mathbb{C},\ $\ calculated from the expressions
\ (\ref{V14}). In particular, \ the compatibility condition \ (\ref{V16}) for
the first two equations of \ (\ref{V22a}) is equivalent to the Dunajski metric
nonlinear equations\ (\ref{V21}).\ 

The description of the Lax--Sato equations presented above, especially their
alternative differential-geometric interpretation \ (\ref{V16a}) and
\ (\ref{V16b}), \ makes it possible to realize\ that the structure group
$Diff_{hol}(\mathbb{C\times}\mathbb{T}^{n})$ should play an important role in
unveiling the hidden Lie-algebraic nature of the integrable heavenly dynamical
systems. This is actually the case, and a detailed analysis is presented in
the sequel.

\section{\label{Sec_2}Heavenly equations: the Lie-algebraic integrability
scheme}

Let $\tilde{G}_{\pm}:=\widetilde{Diff}_{\pm}(\mathbb{T}^{n}),$ $n\in
\mathbb{Z}_{+},$ be subgroups of the loop diffeomorphisms group $\widetilde
{Diff}(\mathbb{T}^{n}):=$ $\{\mathbb{C}\supset\mathbb{S}^{1}\rightarrow
Diff(\mathbb{T}^{n})\},$ holomorphically extended in the interior
$\mathbb{S}_{+}^{1}\subset\mathbb{C}$ and in the exterior $\mathbb{S}_{-}%
^{1}\subset\mathbb{C}$ \ regions of the unit circle $\mathbb{S}^{1}%
\subset\mathbb{C}^{1},$ such that for any $g(\lambda)\in\tilde{G}_{\pm},$
$\lambda\in\mathbb{\ }\mathbb{S}_{-}^{1},$ $g(\infty)=1\in Diff(\mathbb{T}%
^{n}).$ The corresponding Lie subalgebras $\mathcal{\tilde{G}}_{\pm
}:=\widetilde{diff}_{\pm}(\mathbb{T}^{n})$ of the loop subgroups $\ \tilde
{G}_{\pm}$ are vector fields on $\mathbb{T}^{n}$ \ holomorphic, respectively,
on $\mathbb{S}_{\pm}^{1}\subset\mathbb{C}^{1},$ where for any $\tilde
{a}(\lambda)\in\mathcal{\tilde{G}}_{-}$ \ the value $\tilde{a}(\infty)=0.$
\ The split loop Lie algebra $\mathcal{\tilde{G}=\tilde{G}}_{+}+$
$\mathcal{\tilde{G}}_{-}$ can be naturally identified with a dense subspace of
the dual space $\mathcal{\tilde{G}}^{\ast}$ through the pairing
\begin{equation}
(\tilde{l},\tilde{a}):=\frac{1}{2\pi i}\oint\limits_{\mathbb{S}^{1}%
}(l(x,\lambda),a(x,\lambda))_{H^{q}}\frac{d\lambda}{\lambda^{p}},
\label{eq1.1}%
\end{equation}
for\ some fixed $p,q\in\mathbb{Z}_{+}.\ \ $ We took above, by definition
\cite{Cart,PrPr}, a loop vector field $\tilde{a}\in\Gamma(\tilde{T}%
(\mathbb{T}^{n}))$\ and a loop differential 1-form $\tilde{l}\in\tilde
{\Lambda}^{1}(\mathbb{T}^{n})$ given as
\begin{align}
&  \tilde{a}\ =\sum\limits_{j=1}^{n}a^{(j)}(x,\lambda)\frac{\partial}{\partial
x_{j}}:=\left\langle a(x;\lambda),\frac{\partial}{\partial x}\right\rangle
,\label{eq1.2}\\
&  \tilde{l}\ =\sum\limits_{j=1}^{n}l_{j}(x,\lambda)dx_{j}:=\ \left\langle
l(x;\lambda),dx\right\rangle ,\nonumber
\end{align}
introduced for brevity the gradient operator $\frac{\partial}{\partial
x}:=\left(  \frac{\partial}{\partial x_{1}},\frac{\partial}{\partial x_{2}%
},...,\frac{\partial}{\partial x_{n}}\right)  ^{\intercal}$ in the Euclidean
space $\mathbb{E}^{n}$ \ and chose the Sobolev type metric $(\cdot
,\cdot)_{H^{q}}$ on the space $\ C^{\infty}(\mathbb{T}^{n};\mathbb{R}%
^{n})\subset H^{q}(\mathbb{T}^{n};\mathbb{R}^{n})$ for some $q\in
\mathbb{Z}_{+}$ as
\begin{equation}
(l(x;\lambda),a(x;\lambda))_{H^{q}}:=\sum\limits_{j=1}^{n}\sum\limits_{|\alpha
|=0}^{q}\int\limits_{\mathbb{T}^{n}}dx\left(  \frac{\partial^{|\alpha|}%
l_{j}(x;\lambda)}{\partial x^{\alpha}}\frac{\partial^{|\alpha|}a^{(j)}%
(x;\lambda)}{\partial x^{\alpha}}\right)  , \label{eq1.2a}%
\end{equation}
where $\partial x^{\alpha}:=\partial x_{1}^{\alpha_{1}}\partial x_{2}%
^{\alpha_{2}}...\partial x_{2}^{\alpha_{n}},|\alpha|=\sum_{j=1}^{n}\alpha_{j}$
for $\alpha\in\mathbb{Z}_{+}^{n},$ \ \ generalizing the metric used before in
\cite{Misi}. \ The Lie commutator of vector fields $\tilde{a},\tilde{b}$
$\in\mathcal{\tilde{G}} $ \ is calculated the standard way and equals
\begin{align}
\lbrack\tilde{a},\tilde{b}]  &  =\tilde{a}\tilde{b}-\tilde{b}\tilde
{a}=\left\langle \left\langle a(x;\lambda),\frac{\partial}{\partial
x}\right\rangle b(x;\lambda),\frac{\partial}{\partial x}\right\rangle
-\label{eq1.3}\\
& \nonumber\\
&  -\ \left\langle \left\langle b(x;\lambda),\frac{\partial}{\partial
x}\right\rangle a(x;\lambda),\frac{\partial}{\partial x}\right\rangle
.\nonumber
\end{align}
The Lie algebra $\mathcal{\tilde{G}}$ \ naturally splits into the direct sum
of two Lie subalgebras
\begin{equation}
\mathcal{\tilde{G}}=\mathcal{\tilde{G}}_{+}\oplus\mathcal{\tilde{G}}_{-},
\label{eq1.3a}%
\end{equation}
for which one can identify the dual spaces
\[
\mathcal{\tilde{G}}_{+}^{\ast}\simeq\lambda^{p-1}\mathcal{\tilde{G}}%
_{-},\ \ \ \ \ \ \mathcal{\tilde{G}}_{-}^{\ast}\simeq\lambda^{p-1}%
\mathcal{\tilde{G}}_{+},
\]
where for any $l(\lambda)\in\mathcal{\tilde{G}}_{-}^{\ast}$ one has the
constraint $\tilde{l}(0)=0$. Having defined now the projections
\begin{equation}
P_{\pm}\mathcal{\tilde{G}}:=\mathcal{\tilde{G}}_{\pm}\subset\mathcal{\tilde
{G}}, \label{eq1.4}%
\end{equation}
one can construct a classical $\mathcal{R}$-structure \cite{FaTa,ReSe,Seme} on
the Lie algebra $\mathcal{\tilde{G}}$ as the endomorphism $\mathcal{R}%
:\mathcal{\tilde{G}\rightarrow\tilde{G}},$ where
\begin{equation}
\mathcal{R}:=\ (P_{+}-P_{-})/2, \label{eq1.5}%
\end{equation}
which allows to determine on the vector space $\mathcal{\tilde{G}}$ \ the new
Lie algebra structure
\begin{equation}
\lbrack\tilde{a},\tilde{b}]_{\mathcal{R}}:=[\mathcal{R}\tilde{a},\tilde
{b}]+[\tilde{a},\mathcal{R}\tilde{b}] \label{eq1.6}%
\end{equation}
for any $\tilde{a},\tilde{b}\in\mathcal{\tilde{G}},$ \ satisfying the standard
Jacobi identity.

Let $\mathrm{D}\mathcal{(\tilde{G}}^{\ast})$ denote the space of smooth
functions on $\mathcal{\tilde{G}}^{\ast}$. Then for any $f,g\in
\mathcal{D(\tilde{G}}^{\ast})$ one can write the canonical
\cite{FaTa,ReSe,PrMy,BlPrSa} Lie--Poisson bracket
\begin{equation}
\{f,g\}:=(\tilde{l},[\nabla f(\tilde{l}),\nabla g(\tilde{l})]), \label{eq1.8}%
\end{equation}
where $\tilde{l}\in\mathcal{\tilde{G}}^{\ast}$ is a seed element and $\nabla
f,$ $\nabla g\in\mathcal{\tilde{G}}$ \ are the standard functional gradients
at $\tilde{l}\in\mathcal{\tilde{G}}^{\ast}$ with respect to the metric
(\ref{eq1.1}). The related to (\ref{eq1.8}) space $I\mathcal{(\tilde{G}}%
^{\ast})$ of \ Casimir invariants is defined as the\ set $I(\mathcal{\tilde
{G}}^{\ast})\subset\mathrm{D}\mathcal{(\tilde{G}}^{\ast})$ of smooth
independent functions $\gamma_{j}\in\mathrm{D}\mathcal{(\tilde{G}}^{\ast
}),j=\overline{1,n},$ for which
\begin{equation}
ad_{\nabla h_{j}(\tilde{l})}^{\ast}\tilde{l}=0, \label{eq1.9}%
\end{equation}
where for any seed element
\begin{equation}
\tilde{l}=<l,dx>\ \ \label{eq1.10}%
\end{equation}
the gradients%
\begin{equation}
\nabla\gamma_{j}(\tilde{l}):=\ \left\langle \nabla\gamma_{j}(l),\frac
{\partial}{\partial x}\right\rangle \ \ \label{eq1.11}%
\end{equation}
and the coadjoint action (\ref{eq1.9}) can be equivalently rewritten, for
instance, in the case $q=0,$ as
\begin{equation}
\ \left\langle \frac{\partial}{\partial x},\nabla\gamma_{j}(l)\right\rangle
l+\ \left\langle l,(\frac{\partial}{\partial x}\nabla\gamma_{j}%
(l))\right\rangle =0\ \label{eq1.12}%
\end{equation}
for any $j=\overline{1,n}.$ If one takes two smooth functions $\gamma
^{(y)},\gamma^{(t)}\in I(\mathcal{\tilde{G}}^{\ast})\subset\mathrm{D}%
(\mathcal{\tilde{G}}^{\ast}),$ \ their second Poisson bracket
\begin{equation}
\{\gamma^{(y)},\gamma^{(t)}\}_{\mathcal{R}}:=(\tilde{l},[\nabla\gamma
^{(y)},\nabla\gamma^{(t)}]_{\mathcal{R}}) \label{eq1.13}%
\end{equation}
on the space $\mathcal{\tilde{G}}^{\ast}$ vanishes; that is,
\begin{equation}
\{\gamma^{(y)},\gamma^{(t)}\}_{\mathcal{R}}=0 \label{eq1.13a}%
\end{equation}
at any seed element $\tilde{l}\in\mathcal{\tilde{G}}^{\ast}.$ Since the
functions $\gamma^{(y)},\gamma^{(t)}\in I(\mathcal{\tilde{G}}^{\ast}),$ the
following coadjoint action relationships hold:
\begin{equation}
ad_{\nabla\gamma^{(y)}(\tilde{l})}^{\ast}\tilde{l}=0,\ \ \ \ \ ad_{\nabla
\gamma^{(t)}(\tilde{l})}^{\ast}\tilde{l}=0, \label{eq1.14}%
\end{equation}
which can be equivalently rewritten (as above in the case $q=0)$ as
\begin{equation}%
\begin{array}
[c]{c}%
\left\langle \frac{\partial}{\partial x},\nabla\gamma^{(y)}(l)\right\rangle
l+\ \left\langle l,(\frac{\partial}{\partial x}\nabla\gamma^{(y)}%
(l))\right\rangle \ =\\
\\
=\left\langle \nabla\gamma^{(y)}(l),\frac{\partial}{\partial x}\right\rangle
l+\left\langle (\frac{\partial}{\partial x},\nabla\gamma^{(y)}%
(l))\right\rangle l+\left\langle l,(\frac{\partial}{\partial x}\nabla
\gamma^{(y)}(l))\right\rangle :=\\
\\
=(\tilde{A}_{\nabla\gamma^{(y)}}+B_{\nabla\gamma^{(y)}})l\
\end{array}
\ \label{eq1.15}%
\end{equation}
and similarly
\begin{equation}
\
\begin{array}
[c]{c}%
\left\langle \frac{\partial}{\partial x},\nabla\gamma^{(t)}(l)\right\rangle
l+\ \left\langle l,(\frac{\partial}{\partial x},\nabla\gamma^{(t)}%
(l))\right\rangle \ :=(\tilde{A}_{\nabla\gamma^{(t)}}+B_{\nabla\gamma^{(t)}%
})l,
\end{array}
\label{eq1.16}%
\end{equation}
where the expressions
\begin{equation}
\tilde{A}_{\nabla\gamma^{(y)}}:=\ \left\langle \nabla\gamma^{(y)}(l),\frac
{d}{dx}\right\rangle ,\text{ \ \ }\tilde{A}_{\nabla\gamma^{(t)}}:=\left\langle
\nabla\gamma^{(t)}(l),\frac{\partial}{\partial x}\right\rangle \label{eq1.16a}%
\end{equation}
are true vector fields on $\mathbb{T}^{n},$ yet the expressions
\begin{align}
&  B_{\nabla\gamma^{(y)}}:=\ \left\langle (\frac{\partial}{\partial x}%
,\nabla\gamma^{(y)}(l))\right\rangle \ +\ \left(  \frac{\partial}{\partial
x}\nabla\gamma^{(y)}(l)\right)  ,\label{eq1.17}\\
& \nonumber\\
&  B_{\nabla\gamma^{(t)}}:=\ \left\langle (\frac{\partial}{\partial x}%
,\nabla\gamma^{(t)}(l))\right\rangle \ +\ \left(  \frac{\partial}{\partial
x}\nabla\gamma^{(t)}(l)\right)  ,\nonumber
\end{align}
are the usual matrix homomorphisms of the Euclidean space $\mathbb{E}^{n}.$

Consider now the following Hamiltonian flows on the space $\mathcal{\tilde{G}%
}^{\ast}$:
\begin{align}
&  \partial\tilde{l}/\partial y:=\{h^{(y)},\tilde{l}\}_{\mathcal{R}%
}=-ad_{\nabla h^{(y)}(\tilde{l})_{+}}^{\ast}\tilde{l},\label{eq1.18}\\
& \nonumber\\
&  \partial\tilde{l}/\partial t:=\{h^{(t)},\tilde{l}\}_{\mathcal{R}%
}=-ad_{\nabla h^{(t)}(\tilde{l})_{+}}^{\ast}\tilde{l},\nonumber
\end{align}
where $h^{(y)},h^{(t)}\in I(\mathcal{\tilde{G}}^{\ast})$ and $y,t\in
\mathbb{R}$ are the corresponding evolution parameters. \ Since $h^{(y)}%
,h^{(t)}\in I(\mathcal{\tilde{G}}^{\ast})$ are Casimirs, the flows
(\ref{eq1.18}) commute. Thus, taking into account the representations
(\ref{eq1.15}), one can recast the flows (\ref{eq1.18}) as
\begin{equation}
\partial l/\partial t=-(\tilde{A}_{\nabla h_{+}^{(t)}}+B_{\nabla h_{+}^{(t)}%
})l,\ \ \ \ \ \partial l/\partial y=-(\tilde{A}_{\nabla h_{+}^{(y)}}+B_{\nabla
h_{+}^{(y)}})l, \label{eq1.19}%
\end{equation}
where%
\begin{equation}
\tilde{A}_{\nabla h_{+}^{(t)}}:=\ \left\langle \nabla h^{(t)}(l)_{+}%
,\frac{\partial}{\partial x}\right\rangle ,\ \ \ \ \ \tilde{A}_{\nabla
h_{+}^{(y)}}:=\ \left\langle \nabla h^{(y)}(l)_{+},\frac{\partial}{\partial
x}\right\rangle . \label{eq1.19a}%
\end{equation}

\begin{lemma}
\bigskip\label{Lm_1.1}The compatibility of commuting flows \ (\ref{eq1.19}) is
equivalent to the Lax type vector fields relationship
\begin{equation}
\partial\tilde{A}_{\nabla h_{+}^{(t)}}/\partial y-\partial\tilde{A}_{\nabla
h_{+}^{(y)}}/\partial t+[\tilde{A}_{\nabla h_{+}^{(y)}},\tilde{A}_{\nabla
h_{+}^{(t)}}]=0, \label{eq1.20}%
\end{equation}
which holds for all $y,t\in\mathbb{R}$ \ and arbitrary $\lambda\in\mathbb{C}.
$
\end{lemma}

\begin{proof}
The compatibility of commuting flows (\ref{eq1.19}) implies that $\partial
^{2}l/\partial t\partial y-\partial^{2}l/\partial y\partial t=0$ \ for all
$y,t\in\mathbb{R}$ \ and arbitrary $\lambda\in\mathbb{C}.$ Taking into account
the expressions \ (\ref{eq1.18}), \ one has for any vector field $\tilde
{Z}=<Z,\frac{\partial}{\partial x}>\in\mathcal{\tilde{G}}$
\[%
\begin{array}
[c]{c}%
0=(\partial^{2}\tilde{l}/\partial t\partial y-\partial^{2}\tilde{l}/\partial
y\partial t,\tilde{Z})=-\frac{\partial}{\partial t}(ad_{\nabla h^{(y)}%
(\tilde{l})_{+}}^{\ast}\tilde{l},\tilde{Z})+\frac{\partial}{\partial
y}(ad_{\nabla h^{(t)}(\tilde{l})_{+}}^{\ast}\tilde{l},\tilde{Z})=\\
\\
=-\frac{\partial}{\partial t}(\ \tilde{l},[\nabla h^{(y)}(\tilde{l}%
)_{+},\tilde{Z}])+\frac{\partial}{\partial y}(\ \tilde{l},[\nabla
h^{(y)}(\tilde{l})_{+},\tilde{Z}])=\\
\\
=-(\ \frac{\partial}{\partial t}\tilde{l},[\nabla h^{(y)}(\tilde{l}%
)_{+},\tilde{Z}])-(\ \tilde{l},[\frac{\partial}{\partial t}\nabla
h^{(y)}(\tilde{l})_{+},\tilde{Z}])+\\
\\
+(\ \frac{\partial}{\partial y}\tilde{l},[\nabla h^{(t)}(\tilde{l})_{+}%
,\tilde{Z}])+(\ \tilde{l},[\frac{\partial}{\partial y}\nabla h^{(y)}(\tilde
{l})_{+},\tilde{Z}])=\\
\\
=(\ ad_{\nabla h^{(t)}(\tilde{l})_{+}}^{\ast}\tilde{l},[\nabla h^{(y)}%
(\tilde{l})_{+},\tilde{Z}])-(\ \tilde{l},[\frac{\partial}{\partial t}\nabla
h^{(y)}(\tilde{l})_{+},\tilde{Z}])-\\
\\
-(\ ad_{\nabla h^{(y)}(\tilde{l})_{+}}^{\ast}\tilde{l},[\nabla h^{(t)}%
(\tilde{l})_{+},\tilde{Z}])+(\ \tilde{l},[\frac{\partial}{\partial y}\nabla
h^{(y)}(\tilde{l})_{+},\tilde{Z}])=
\end{array}
\]%
\begin{equation}%
\begin{array}
[c]{c}%
\\
=(\ \tilde{l},[\nabla h^{(t)}(\tilde{l})_{+},[\nabla h^{(y)}(\tilde{l}%
)_{+},\tilde{Z}]])-(\ \tilde{l},[\frac{\partial}{\partial t}\nabla
h^{(y)}(\tilde{l})_{+},\tilde{Z}])-\\
\\
-(\tilde{l},[\nabla h^{(y)}(\tilde{l})_{+},[\nabla h^{(t)}(\tilde{l}%
)_{+},\tilde{Z}]])+(\ \tilde{l},[\frac{\partial}{\partial y}\nabla
h^{(t)}(\tilde{l})_{+},\tilde{Z}])=\\
\\
=(\ \tilde{l},[\nabla h^{(t)}(\tilde{l})_{+},[\nabla h^{(y)}(\tilde{l}%
)_{+},\tilde{Z}]]-[\nabla h^{(y)}(\tilde{l})_{+},[\nabla h^{(t)}(\tilde
{l})_{+},\tilde{Z}]])+\\
\\
+(\ \tilde{l},[\frac{\partial}{\partial y}\nabla h^{(t)}(\tilde{l})_{+}%
-\frac{\partial}{\partial t}\nabla h^{(y)}(\tilde{l})_{+},\tilde{Z}])=\\
\\
=(\ \tilde{l},[\nabla h^{(y)}(\tilde{l})_{+},\nabla h^{(t)}(\tilde{l}%
)_{+}]+\frac{\partial}{\partial y}\nabla h^{(t)}(\tilde{l})_{+}-\frac
{\partial}{\partial t}\nabla h^{(y)}(\tilde{l})_{+},\tilde{Z})=\\
\\
=(\tilde{l},[\partial\tilde{A}_{\nabla h_{+}^{(t)}}/\partial y-\partial
\tilde{A}_{\nabla h_{+}^{(y)}}/\partial t+[\tilde{A}_{\nabla h_{+}^{(y)}%
},\tilde{A}_{\nabla h_{+}^{(t)}}],Z])=(ad_{\nabla\tilde{H}}^{\ast}\tilde
{l},\tilde{Z}),
\end{array}
\label{eq1.21}%
\end{equation}
where
\begin{equation}
\ \nabla H(\tilde{l}):=\partial\tilde{A}_{\nabla h_{+}^{(t)}}/\partial
y-\partial\tilde{A}_{\nabla h_{+}^{(y)}}/\partial t+[\tilde{A}_{\nabla
h_{+}^{(y)}},\tilde{A}_{\nabla h_{+}^{(t)}}]. \label{eq1.21a}%
\end{equation}
From \ (\ref{eq1.21}) we obtain that $ad_{\nabla H(\tilde{l})}^{\ast}\tilde
{l}=0$ for all $y,t\in\mathbb{R}$ \ and arbitrary $\lambda\in\mathbb{C}$; that
i,s the vector field \ (\ref{eq1.21a}) is the gradient of an analytical
Casimir functional $H\in I(\mathcal{\tilde{G}}^{\ast}).$ Now based on the
analyticity of the vector field expression \ (\ref{eq1.21a}), one easily shows
\cite{BoDrMa} that $\nabla H(\tilde{l})=0,$ thus finishing the proof.
\end{proof}

For the exact representatives of the functions $h^{(y)},h^{(t)}\in
I(\mathcal{\tilde{G}}^{\ast})$, it is necessary to solve the determining
equation (\ref{eq1.12}), taking into account that if the chosen element
$\tilde{l}\in\mathcal{\tilde{G}}^{\ast}$ is singular as $|\lambda
|\rightarrow\infty,$ the related expansion
\begin{equation}
\nabla\gamma^{(p)}(l)\sim\lambda^{p}\sum\limits_{j\in\mathbb{Z}_{+}}%
\nabla\gamma_{j}(l)\lambda^{-j}, \label{eq1.22}%
\end{equation}
where the degree $p\in\mathbb{Z}_{+}$ can be taken as arbitrary. Upon
substituting (\ref{eq1.22}) into (\ref{eq1.12}) one can find recurrently all
the coefficients $\nabla\gamma_{j}(l),$ $j\in\mathbb{Z}_{+},$ and then
construct gradients of the Casimir functions $h^{(y)},h^{(t)}\in
I(\mathcal{\tilde{G}}^{\ast})$ reduced on $\mathcal{\tilde{G}}_{+\text{ \ }}%
$as
\begin{equation}
\nabla h_{+}^{(t)}(l)\ =(\lambda^{p_{t}}\nabla\gamma(l))|_{+},\ \ \ \ \ \nabla
h_{+}^{(y)}(l)=(\lambda^{p_{y}}\nabla\gamma(l))|_{+} \label{eq1.23}%
\end{equation}
for some positive integers $p_{y},p_{t}\in\mathbb{Z}_{+}.$

\begin{remark}
\label{Rem1.1} \textit{\ As mentioned above, the expansion (\ref{eq1.22}) is
effective if a chosen seed element } $\tilde{l}\in\mathcal{\tilde{G}}^{\ast}%
$\textit{\ is singular as} $|\lambda|\rightarrow\infty.$ \textit{In the case
when it is singular as} $|\lambda|\rightarrow0,$ \textit{the expression
(\ref{eq1.22}) should be replaced by the expansion}%
\begin{equation}
\nabla\gamma^{(p)}(l)\sim\lambda^{-p}\sum\limits_{j\in\mathbb{Z}_{+}}%
\nabla\gamma_{j}(l)\lambda^{j}, \label{eq1.22a}%
\end{equation}
\textit{for an arbitrary }$p\in\mathbb{Z}_{+},$ \ \textit{and the reduced
Casimir function gradients then are given by the expressions
\begin{equation}
\nabla h_{-}^{(y)}(l)=\lambda(\lambda^{-p_{y}-1}\nabla\gamma(l))|_{-}%
,\ \ \ \ \ \nabla h_{-}^{(t)}(l)=\lambda(\lambda^{-p_{t}-1}\nabla
\gamma(l))|_{-}\label{eq1.23a}%
\end{equation}
\textit{f}or some positive integers }$p_{y},p_{t}\in\mathbb{Z}_{+}$ $.$ Then
\textit{the corresponding flows are, respectively, written as
\begin{equation}
\partial\tilde{l}/\partial t=ad_{\triangledown h_{-}^{(t)}(\tilde{l})}^{\ast
}\tilde{l},\text{ \ \ }\partial\tilde{l}/\partial y=ad_{\triangledown
h_{-}^{(y)}(\tilde{l})}^{\ast}\tilde{l}.\label{eq1.18a}%
\end{equation}
}
\end{remark}

The above results, owing to \ \ref{Lm_1.1}, can be formulated as the following
main proposition.

\begin{proposition}
\label{prop1.1} Let a seed vector field $\tilde{l}\in\mathcal{\tilde{G}}%
^{\ast}$ and $h^{(y)},h^{(t)}\in I(\mathcal{\tilde{G}}^{\ast})$ be Casimir
functions subject to the metric $(\cdot,\cdot)$ on the loop Lie algebra
$\mathcal{\tilde{G}}$ \ and the natural coadjoint action on the loop
co-algebra $\mathcal{\tilde{G}}^{\ast}.$ Then the following dynamical systems
\begin{equation}
\partial\tilde{l}/\partial y=-ad_{\nabla h_{+}^{(y)}(\tilde{l})}^{\ast}%
\tilde{l},\text{ \ \ }\partial\tilde{l}/\partial t=-ad_{\nabla h_{+}%
^{(t)}(\tilde{l})}^{\ast}\tilde{l} \label{eq1.23b}%
\end{equation}
are commuting Hamiltonian flows for all $y,t\in\mathbb{R}.$ \ Moreover, the
compatibility condition of these flows is equivalent to the so called vector
fields representation
\begin{equation}
(\partial/\partial t+\tilde{A}_{\nabla h_{+}^{(t)}})\psi=0,\ \ \ \ \ (\partial
/\partial y+\tilde{A}_{\nabla h_{+}^{(y)}})\psi=0, \label{eq1.23c}%
\end{equation}
where $\psi\in C^{\infty}(\mathbb{R}^{2}\times\mathbb{T}^{n};\mathbb{C})$ and
the vector fields $\tilde{A}_{\nabla h_{+}^{(y)}},\tilde{A}_{\nabla
h_{+}^{(t)}}\in\mathcal{\tilde{G}},$ \ given by the expressions (\ref{eq1.19a}%
) and (\ref{eq1.23}), satisfy the Lax relationship \ (\ref{eq1.20}).
\end{proposition}

The proposition above makes it possible to describe in a very effective way
the B\"{a}cklund transformations \ between two solution sets to
the\ dispersionless heavenly type equations resulting from the Lax
compatibility condition \ (\ref{eq1.20}). Namely, let a diffeomorphism $\xi
\in\widetilde{Diff}(\mathbb{T}^{n}),$ depending parametrically on $\lambda
,\mu\in\mathbb{C}$ and evolution variables $(y,t)\in\mathbb{R}^{2},$ be such
that a seed loop differential form $\tilde{l}(x;\lambda,\mu)\in\mathcal{\tilde
{G}}^{\ast}\simeq\tilde{\Lambda}^{1}(\mathbb{T}^{n})$ satisfies the invariance
condition
\begin{equation}
\ \tilde{l}(\xi(x;\lambda,\mu);\lambda)=k\tilde{l}(x;\mu) \label{eq1.23d}%
\end{equation}
for\ some\ non-zero constant $k\in\mathbb{C}\backslash\{0\},$ any
$x\in\mathbb{T}^{n}\ $and arbitrarily chosen $\lambda\in\mathbb{C}.$ As the
seed element $\tilde{l}(\xi(x;\lambda,\mu);\lambda)\ \in\tilde{\Lambda}%
^{1}(\mathbb{T}^{n}),$ by the construction, simultaneously satisfies the
system of compatible equations \ following from (\ref{eq1.23b}), the loop
diffeomorphism $\xi\in\widetilde{Diff}(\mathbb{T}^{n}),$ found analytically
from the invariance condition \ (\ref{eq1.23d}), should satisfiy the
relationships%
\begin{equation}
\frac{\partial}{\partial y}\xi=\nabla h_{+}^{(y)}(l),\text{ \ \ }%
\frac{\partial}{\partial t}\xi=\nabla h_{+}^{(t)}(l), \label{eq1.23e}%
\end{equation}
giving rise exactly to the B\"{a}cklund type relationships for coefficients of
the seed\ loop differential form $\tilde{l}\in\mathcal{\tilde{G}}^{\ast}%
\simeq\tilde{\Lambda}^{1}(\mathbb{T}^{n}).$

\section{\label{Sec_3}Integrable heavenly equations: Examples}

\subsection{ The Mikhalev--Pavlov heavenly equation}

This equation \cite{Mikh,Pavl} is
\begin{equation}
u_{xt}+u_{yy}=u_{y}u_{xx}-u_{x}u_{xy}, \label{eq1.24}%
\end{equation}
where $u\in C^{\infty}(\mathbb{R}^{2}\times\mathbb{T}^{1};\mathbb{R})$ and
$(t,y;x)\in\mathbb{R}^{2}\times\mathbb{T}^{1}.$ Set $\mathcal{\tilde{G}}%
^{\ast}:=\widetilde{diff}^{\ast}(\mathbb{T}^{1})$ and take the corresponding
seed element $\ \tilde{l}\in\mathcal{\tilde{G}}^{\ast}$ as
\begin{equation}
\tilde{l}=(\lambda-2u_{x})dx\ . \label{eq1.25}%
\end{equation}
It generates a Casimir invariant $h\in I(\mathcal{\tilde{G}}^{\ast})$ for
which the expansion (\ref{eq1.22}) as $|\lambda|\rightarrow\infty$ \ is given
by the asymptotic series
\begin{equation}
\nabla h(l)\sim1+u_{x}/\lambda-u_{y}/\lambda^{2}+O(1/\lambda^{3})
\label{eq1.26}%
\end{equation}
and so on. If further one defines
\begin{align}
&  \nabla h^{(t)}(l)_{+}:=(\lambda^{2}\nabla h)_{+}=\lambda^{2}+\lambda
u_{x}-u_{y},\label{eq1.27}\\
& \nonumber\\
&  \nabla h^{(y)}(l)_{+}:=(\lambda^{1}\nabla h)_{+}=\lambda+u_{x},\nonumber
\end{align}
it is easy to verify that
\begin{align}
\tilde{A}_{\nabla h_{+}^{(t)}}  &  :=<\nabla h^{(t)}(l)_{+},\frac{\partial
}{\partial x}>=(\lambda^{2}+\lambda u_{x}-u_{y})\frac{\partial}{\partial
x},\ \ \ \ \ \nonumber\\
& \label{eq1.27a}\\
\tilde{A}_{\nabla h_{+}^{(y)}}  &  :=<\nabla h^{(y)}(l)_{+},\frac{\partial
}{\partial x}>=(\lambda+u_{x})\frac{\partial}{\partial x}.\nonumber
\end{align}
As a result of \ (\ref{eq1.27a}) and the commuting flows \ (\ref{eq1.23b}) on
$\mathcal{\tilde{G}}^{\ast}$ we \ retrieve (the equivalent to the
Mikhalev--Pavlov \cite{Pavl} )\ equation (\ref{eq1.24}) vector field
compatibility relationships
\begin{equation}
\frac{\partial\psi}{\partial t}+(\lambda^{2}+\lambda u_{x}-u_{y}%
)\frac{\partial\psi}{\partial x}=0=\frac{\partial\psi}{\partial y}%
+(\lambda+u_{x})\frac{\partial\psi}{\partial x}\ , \label{eq1.28}%
\end{equation}
satisfied for $\psi\in C^{\infty}(\mathbb{R}^{2}\times\mathbb{T}%
^{1};\mathbb{C}),$ any $(y,t;x)\in\mathbb{R}^{2}\mathbb{\times T}^{1}$ and all
$\lambda\in\mathbb{C}.$

We now study the B\"{a}cklund transformation for two special solutions
$\ u,\tilde{u}\in C^{2}(\mathbb{T}^{1}\times\mathbb{R}^{2};\mathbb{R})$ to the
Mikhalev--Pavlov equation \ (\ref{eq1.24}). Let us consider a loop
diffeomorphism $\xi\in\widetilde{Diff}(\mathbb{T}^{1})$ that is the mapping
$\mathbb{T}^{1}\ni x\rightarrow\tilde{x}=\xi(x;y,t,\lambda)\in\mathbb{T}^{1},$
which parametrically depends on $\lambda\in\mathbb{C}^{1}$ and the evolution
variables $(y,t)\in\mathbb{R}^{2},$ satisfying the invariance condition
\ (\ref{eq1.23d}) for the seed loop differential form \ (\ref{eq1.25}):
\begin{equation}
\ (\lambda-2\tilde{u}_{\tilde{x}}(\tilde{x};t,y))d\tilde{x}=(\lambda
-2u_{x}(x;t,y))dx, \label{eq1.28a}%
\end{equation}
where for simplicity, we define $\mu=\lambda\in\mathbb{C}$ and the constant
parameter $k=1.$ From \ (\ref{eq1.28a}) one easily finds that
\[
\lambda\xi_{x}(x;t,y)=2[\tilde{u}(\tilde{x};t,y)-u(x;t,y)]_{x}+\lambda,
\]
or, equivalently,
\begin{equation}
\xi(x;\lambda)=x+2(\tilde{u}-u)/\lambda+\alpha(y,t;\lambda), \label{eq1.28b}%
\end{equation}
where $\ \ \lambda\in\mathbb{C}\backslash\{0\}$ and $\ \alpha\in
C^{2}(\mathbb{R}^{2};\mathbb{R})$ is some mapping. \ As the loop
diffeomorphism \ (\ref{eq1.28b}) should simultaneously satisfy the vector
field equations (\ref{eq1.23e}), giving rise (at $\alpha(y,t;\lambda)=0)$ to
the following B\"{a}cklund type differential relationships
\begin{equation}
\ 2(\tilde{u}_{y}-u_{y})/\lambda=\lambda+\tilde{u}_{\tilde{x}},\text{
\ \ \ \ \ }2(\tilde{u}_{t}-u_{t})/\lambda=\lambda^{2}+\lambda\tilde{u}%
_{\tilde{x}}-\tilde{u}_{y},\nonumber
\end{equation}
which hold for any $\lambda\in\mathbb{C}$ and two special solutions
$u,\tilde{u}\in C^{2}(\mathbb{T}^{1}\times\mathbb{R}^{2};\mathbb{R})$ to the
Mikhalev--Pavlov equation \ (\ref{eq1.24}).

\subsection{Example: The Witham heavenly type equation}

Consider the following \cite{FeMoSo,Pavl-2,Moro,Kric} heavenly type equation:%
\begin{equation}
u_{ty}=u_{x}u_{xy}-u_{y}u_{xx}, \label{E1}%
\end{equation}
where $u\in C^{2}(\mathbb{R}^{2}\times\mathbb{T}^{1};\mathbb{R})$ and
$(t,y;x)\in\mathbb{R}^{2}\times\mathbb{T}^{1}.$ To prove the Lax-Sato type
integrability of\ (\ref{E1}), let us consider a seed element $\tilde{l}%
\in\tilde{\mathcal{G}}^{\ast},$ defined as
\begin{equation}
\tilde{l}=(u_{y}^{-2}\lambda^{-1}+2u_{x}+\lambda)dx, \label{E2}%
\end{equation}
where $\lambda\in\mathbb{C}\backslash\{0\}$ is a complex parameter. The
following asymptotic expressions are gradients of the Casimir functionals
$h^{(t)},h^{(y)}\in I(\tilde{\mathcal{G}}^{\ast}),$ related with the
holomorphic loop Lie algebra $\tilde{\mathcal{G}}=\widetilde{diff}%
(\mathbb{T}^{1}):$%
\begin{equation}
\nabla h^{(t)}\sim\lambda\lbrack(u_{x}\lambda^{-1}-1)+O(1/\lambda), \label{E3}%
\end{equation}
as $\lambda\rightarrow\infty,$ and
\begin{equation}
\nabla h^{(y)}\sim u_{y}\lambda^{-1}+O(\lambda^{2}), \label{E4}%
\end{equation}
as $\lambda\rightarrow0.$ Based on the expressions (\ref{E3}) and
\ (\ref{E4}), one can construct \cite{PrPr} the following commuting to each
other Hamiltonian flows
\begin{equation}
\frac{\partial}{\partial y}\tilde{l}=-ad_{\nabla h_{-}^{(y)}}^{\ast}\tilde
{l},\text{ \ \ }\frac{\partial}{\partial t}\tilde{l}=-ad_{\nabla h_{+}^{(t)}%
}^{\ast}\tilde{l}\text{\ }\ \text{\ \ \ \ \ \ \ \ \ \ \ \ \ \ \ \ \ \ \ }
\label{E5}%
\end{equation}
with respect to the evolution parameters $y,t\in\mathbb{R},$ which give rise,
in part, to the following equations:%
\begin{align}
&  u_{yt}=u_{x}u_{xy}-u_{y}u_{xx},\label{E6}\\
& \nonumber\\
&  \ \ u_{t}=-u_{y}^{-2}/2+3u_{x}^{2}/2,\nonumber\\
& \nonumber\\
&  u_{yy}=-u_{y}^{3}[(u_{x}u_{y})_{x}+u_{x}u_{xy}],\nonumber
\end{align}
where the projected gradients $\ \nabla h_{-}^{(y)},\nabla h_{+}^{(t)}%
\in\tilde{\mathcal{G}}$ are equal to the loop vector fields
\begin{equation}
\nabla h_{+}^{(t)}=(u_{x}\ -\lambda)\frac{\partial}{\partial x},\text{
\ \ }\nabla h_{-}^{(y)}=\frac{u_{y}}{\lambda}\frac{\partial}{\partial
x},\ \label{E7}%
\end{equation}
satisfying for evolution parameters $y,t\in\mathbb{R}^{2}$ the Lax-Sato vector
field compatibility condition:%
\begin{equation}
\frac{\partial}{\partial y}\nabla h_{+}^{(t)}-\frac{\partial}{\partial
t}\nabla h_{-}^{(y)}+[\nabla h_{+}^{(t)},\nabla h_{-}^{(y)}]=0. \label{E8}%
\end{equation}
As a simple consequence of the condition one finds exactly the first equation
of the \ (\ref{E6}), coinciding with the heavenly type equation (\ref{E1}).
Thereby, we have stated that this equation is a completely integrable heavenly
type dynamical system with respect to both evolution parameters.

\begin{remark}
It is worth to observe that the third equation of \ (\ref{E6}) entails the
interesting relationship%
\begin{equation}
\frac{\partial}{\partial y}(1/u_{y})=\frac{\partial}{\partial x}(u_{x}%
u_{y}^{2}), \label{E9}%
\end{equation}
whose compatibility makes it possible to introduce a new function $v\in
C^{2}(\mathbb{S}^{1};\mathbb{R}),$ satisfying the next differential
expressions:%
\begin{equation}
v_{x}=1/u_{y},\text{ \ \ \ \ }v_{y}=u_{x}u_{y}^{2}, \label{E10}%
\end{equation}
which hold for all $(x,y)\in\mathbb{S}^{1}\times\mathbb{R}.$ Based on
\ (\ref{E10}) the seed element \ (\ref{E2}) is rewritten as
\begin{equation}
\tilde{l}=(v_{x}^{2}\lambda^{-1}+2u_{x}+\lambda)dx, \label{E11}%
\end{equation}
and the vector fields \ (\ref{E7}) are rewritten as
\begin{equation}
\nabla h_{+}^{(t)}=(u_{x}\ -\lambda)\frac{\partial}{\partial x},\text{
\ \ }\nabla h_{-}^{(y)}=\frac{1}{v_{x}\lambda}\frac{\partial}{\partial
x},\ \label{E12}%
\end{equation}
whose compatibility condition \ (\ref{E8}) gives rise to the following system
of heavenly type nonliner integrable flows:%
\begin{align}
v_{y}  &  =u_{x}v_{x}^{-2},\text{ \ }v_{xt}=u_{x}v_{xy}+u_{xx}v_{x,}%
\label{E13}\\
& \nonumber\\
u_{y}  &  =1/v_{x},\text{ \ \ }u_{t}=-v_{x}^{2}/2+3u_{x}^{2}/2,\text{\ }%
\ \nonumber
\end{align}
compatible for arbitrary evolution parameters $y,t\in\mathbb{R}.$
\end{remark}

\subsection{ Pleba\'{n}ski heavenly equation}

This equation \cite{Pleb} is
\begin{equation}
u_{tx_{1}}-u_{yx_{2}}+u_{x_{1}x_{1}}u_{x_{2}x_{2}}-u_{x_{1}x_{2}}^{2}=0
\label{eq1.29}%
\end{equation}
for a function $u\in C^{\infty}(\mathbb{R}^{2};\mathbb{T}^{2}),$ where
$(y,t;x_{1},x_{2})\in\mathbb{R}^{2}\times\mathbb{T}^{2}.$ We set
$\ \mathcal{\tilde{G}}^{\ast}:=\widetilde{diff}^{\ast}(\mathbb{T}^{2})$ and
take the corresponding seed element $\ \tilde{l}\in\mathcal{\tilde{G}}^{\ast}$
as
\begin{equation}
\tilde{l}=(\lambda-u_{x_{1}x_{2}}+u_{x_{1}x_{1}})dx_{1}+(\lambda-u_{x_{2}%
x_{2}}+u_{x_{1}x_{2}})dx_{2}. \label{eq1.30}%
\end{equation}
This generates two independent Casimir functionals $h^{(1)},h^{(2)}\in
I(\mathcal{\tilde{G}}^{\ast}),$ whose gradient expansions (\ref{eq1.22}) as
$|\lambda|\rightarrow\infty$ are given by the expressions
\begin{align}
\nabla h^{(1)}(l)  &  \sim(0,1)^{\intercal}+(u_{x_{2}x_{2}},-u_{x_{1}x_{2}%
})^{\intercal}\lambda^{-1}+O(\lambda^{-2}),\label{eq1.31}\\
& \nonumber\\
\ \ \nabla h^{(2)}(l)  &  \sim(1,0)^{\intercal}+(u_{x_{1}x_{2}},-u_{x_{1}%
x_{2}})^{\intercal}\lambda^{-1}+O(\lambda^{-2}),\nonumber
\end{align}
and so on. Now, by defining
\begin{align}
&  \nabla h^{(y)}(l)_{+}:=(\lambda\nabla h^{(1)}(l))_{+}=(u_{x_{2}x_{2}%
},\lambda-u_{x_{1}x_{2}})^{\intercal},\label{eq1.32}\\
& \nonumber\\
&  \nabla h^{(t)}(l)_{+}:=(\lambda\nabla h^{(2)}(l))_{+}=(\lambda
+u_{x_{1}x_{2}},-u_{x_{1}x_{1}})^{\intercal},\nonumber
\end{align}
one obtains for \ (\ref{eq1.29}) the following \cite{Pleb} vector field
representation
\begin{align}
&  \frac{\partial\psi}{\partial t}+u_{x_{1}x_{1}}\frac{\partial\psi}{\partial
x_{1}}+(\lambda-u_{x_{1}x_{2}})\frac{\partial\psi}{\partial x_{2}%
}=0,\label{eq1.33}\\
& \nonumber\\
&  \frac{\partial\psi}{\partial z}+(\lambda+u_{x_{1}x_{2}})\frac{\partial\psi
}{\partial x_{1}}-u_{x_{1}x_{1}}\frac{\partial\psi}{\partial x_{2}%
}=0,\nonumber
\end{align}
satisfied for $\psi\in C^{\infty}(\mathbb{R}^{2}\times\mathbb{T}%
^{2};\mathbb{C}),$ any $(t,y;x_{1},x_{2})\in\mathbb{R}^{2}\times\mathbb{T}%
^{2}$ and all $\lambda\in\mathbb{C}.$

\subsection{General heavenly equation}

This equation was first suggested and analyzed by Schief in \cite{Schi,Schi-1}%
, where it was shown to be equivalent to the first Pleba\'{n}ski heavenly
equation, and later studied by Doubrov and Ferapontov \cite{DoFe}; it has the
form
\begin{equation}
\alpha u_{yt}u_{x_{1}x_{2}}+\beta u_{tx_{2}}u_{yx_{1}}+\gamma u_{tx_{1}%
}u_{yx_{2}}=0, \label{eq2.31}%
\end{equation}
where $\alpha,\beta$ and $\gamma\in\mathbb{R}$ are arbitrary constants,
satisfying the constraint
\begin{equation}
\alpha+\beta+\gamma=0, \label{eq2.32}%
\end{equation}
and $t,y\in\mathbb{R},(x_{1},x_{2})\in\mathbb{T}^{2}.$ To demonstrate the Lax
integrability of the equation (\ref{eq2.31}) we choose now a seed vector field
$\tilde{l}\in\mathcal{\tilde{G}}^{\ast}:=\widetilde{diff}^{\ast}%
(\mathbb{T}^{2})$ in the following rational form%
\begin{align}
\tilde{l}  &  =\ \left(  \frac{\mu u_{x_{1}x_{2}}^{2}}{\gamma(\mu+\beta
)\ }+\frac{u_{x_{1}x_{2}}^{2}}{\alpha}-\frac{\mu u_{x_{1}x_{2}}^{2}}%
{\ \beta(\mu-\gamma)}\right)  \ dx_{1}+\label{eq2.33}\\
& \nonumber\\
&  +\ \left(  \frac{\mu u_{x_{1}x_{2}}u_{x_{2}x_{2}}}{\gamma(\mu+\beta
)\ }+\frac{u_{x_{1}x_{2}}u_{x_{2}x_{2}}}{\alpha}-\frac{\mu u_{x_{1}x_{2}%
}u_{x_{2}x_{2}}}{\ \beta(\mu-\gamma)}\right)  dx_{2},\nonumber
\end{align}
where $a_{j},b_{j}\in C^{\infty}(\mathbb{T}^{2};\mathbb{R}),$ $j=\overline
{0,1},$ are smooth functions and $\mu\in\mathbb{C}$ is a complex parameter$.$
\ The corresponding equations for independent Casimir invariants $\gamma
^{(j)}\in I(\mathcal{\tilde{G}}^{\ast}),j=\overline{1,2},$ are given with
respect to the standard metric $(\cdot,\cdot)$ \ by the following asymptotic
expansions:
\begin{equation}
\nabla\gamma^{(1)}(l)\sim\sum\limits_{j\in\mathbb{Z}_{+}}\nabla\gamma
_{j}^{(1)}(l)\lambda^{j}, \label{eq2.34}%
\end{equation}
as $\ \mu+\beta=\lambda\ \rightarrow0\ $ and
\begin{equation}
\nabla\gamma^{(2)}(l)\sim\sum\limits_{j\in\mathbb{Z}_{+}}\nabla\gamma
_{j}^{(2)}(l)\lambda^{j}, \label{eq2.35}%
\end{equation}
as $\ \mu-\gamma=\lambda\rightarrow0.$ For the first case (\ref{eq2.34}) one
obtains that
\begin{equation}
\nabla\gamma^{(1)}(l)\sim\ \left(  -\frac{\beta u_{tx_{2}}}{u_{x_{1}x_{2}}%
}+\frac{u_{tx_{2}}}{u_{x_{1}x_{2}}}\lambda,\frac{\beta u_{x_{1}x_{1}}%
}{u_{x_{1}x_{2}}}\right)  ^{\intercal}+O(\lambda^{2}) \label{eq2.36}%
\end{equation}
and for the second one (\ref{eq2.35}) one finds that
\begin{equation}
\nabla\gamma^{(2)}(l)\sim\ \left(  \frac{\gamma u_{yx_{2}}}{u_{x_{1}x_{2}}%
}+\frac{u_{yx_{2}}}{u_{x_{1}x_{2}}}\lambda,-\frac{\gamma u_{x_{1}x_{1}}%
}{u_{x_{1}x_{2}}}\right)  ^{\intercal}+O(\lambda^{2}). \label{eq2.37}%
\end{equation}
Here we took into account that the following two Hamiltonian flows on
$\mathcal{\tilde{G}}^{\ast}$
\begin{equation}
\partial\tilde{l}/\partial y=ad_{\nabla h_{-}^{(y)}(\tilde{l})}^{\ast}%
\tilde{l},\text{ \ \ \ }\partial\tilde{l}/\partial t=ad_{\nabla h_{-}%
^{(t)}(\tilde{l})}^{\ast}\tilde{l}\ \label{eq2.38}%
\end{equation}
with respect to the evolution parameters $y,t\in\mathbb{R}$ hold for the
following conservation laws gradients: $\ $
\begin{align}
&  \nabla h_{-}^{(t)}(l):=\left.  \lambda(\lambda^{-2}\nabla\gamma
^{(1)}(l))_{-}\right\vert _{\lambda=\mu+\beta}=\left(  \frac{\mu u_{tx_{2}}%
}{u_{x_{1}x_{2}}(\mu+\beta)},\frac{\beta u_{tx_{1}}}{u_{x_{1}x_{2}}(\mu
+\beta)}\right)  ^{\intercal},\label{eq2.39}\\
& \nonumber\\
&  \nabla h_{-}^{(y)}(l):=\left.  \lambda(\lambda^{-2}\nabla\gamma
^{(2)}(l))_{-}\right\vert _{\lambda=\mu-\gamma}=\left(  \frac{\mu u_{yx_{2}}%
}{u_{x_{1}x_{2}}(\mu-\gamma)},-\frac{\gamma u_{yx_{1}}}{u_{x_{1}x_{2}}%
(\mu-\gamma)}\right)  ^{\intercal}.\ \nonumber
\end{align}
Owing to the compatibility condition of two commuting flows (\ref{eq2.39}),
one can easily rewrite it as the Lax relationship
\begin{equation}
\partial\tilde{A}^{(y)}/\partial t-\partial\tilde{A}^{(t)}/\partial
y=[\tilde{A}^{(y)},\tilde{A}^{(t)}], \label{eq2.41}%
\end{equation}
where
\begin{align}
&  \tilde{A}^{(t)}:=\ \left\langle \nabla h_{-}^{(t)}(l),\frac{\partial
}{\partial x}\right\rangle =\frac{\mu u_{tx_{2}}}{u_{x_{1}x_{2}}(\mu+\beta
)}\frac{\partial}{\partial x_{1}}+\frac{\beta u_{tx_{1}}}{u_{x_{1}x_{2}}%
(\mu+\beta)}\frac{\partial}{\partial x_{2}},\label{eq2.42}\\
& \nonumber\\
\tilde{A}^{(y)}  &  :=\ \left\langle h_{-}^{(t)}(l),\frac{\partial}{\partial
x}\right\rangle =\frac{\mu u_{yx_{2}}}{u_{x_{1}x_{2}}(\mu-\gamma)}%
\frac{\partial}{\partial x_{1}}-\frac{\gamma u_{yx_{1}}}{u_{x_{1}x_{2}}%
(\mu-\gamma)}\frac{\partial}{\partial x_{2}}.\nonumber
\end{align}
An easy calculation shows that the general heavenly equation (\ref{eq2.31})
follows from the compatibility condition (\ref{eq2.41}), whose equivalent
vector field representation is given as
\begin{align}
&  \frac{\mu u_{tx_{2}}}{u_{x_{1}x_{2}}(\mu+\beta)}\frac{\partial\psi
}{\partial x_{1}}+\frac{\beta u_{tx_{1}}}{u_{x_{1}x_{2}}(\mu+\beta)}%
\frac{\partial\psi}{\partial x_{2}}+\frac{\partial\psi}{\partial
t}=0,\label{eq2.43}\\
& \nonumber\\
&  \frac{\mu u_{yx_{2}}}{u_{x_{1}x_{2}}(\mu-\gamma)}\frac{\partial\psi
}{\partial x_{1}}-\frac{\gamma u_{yx_{1}}}{u_{x_{1}x_{2}}(\mu-\gamma)}%
\frac{\partial\psi}{\partial x_{2}}+\frac{\partial\psi}{\partial y}=0\nonumber
\end{align}
for a function $\psi\in C^{\infty}(\mathbb{R}^{2}\times\mathbb{T}%
^{2};\mathbb{C})$ for all $(y,t;x_{1},x_{2})\in\mathbb{R}^{2}\times
\mathbb{T}^{2}.$

We mention here that the related Backlund transformation for the general
heavenly equation \ (\ref{eq2.31}) was recently constructed both in
\ \cite{Serg} and in \ \cite{ShMaYa}, and can be retrieved from the
differential one-form \ (\ref{eq2.33}).

\subsection{ The Alonso--Shabat heavenly equation}

This equation \cite{AlSh} has the form%

\begin{equation}
u_{yx_{2}}-u_{t}u_{yx_{1}}+u_{y}u_{tx_{1}}=0, \label{eq1.34}%
\end{equation}
where $u\in C^{\infty}(\mathbb{R}^{2}\times\mathbb{T}^{2};\mathbb{R}%
),(y,t)\in\mathbb{R}^{2}$ and $(x_{1},x_{2})\in\mathbb{T}^{2}.$ To prove its
Lax integrability, we define a seed element $\tilde{l}\in\mathcal{\tilde{G}%
}^{\ast}:=\widetilde{diff}^{\ast}(\mathbb{T}^{2})$ of the form
\begin{equation}
\tilde{l}=\ \ z_{x_{1}}^{2}(\lambda+1)dx_{1}+z_{x_{1}}z_{x_{2}}(\lambda
+1)dx_{2}, \label{eq1.35}%
\end{equation}
for a fixed function $\ z\in C^{\infty}(\mathbb{T}^{2};\mathbb{R}).$ \ Then
one easily obtains asymptotic expansionsas $|\lambda|\rightarrow\infty$ for
coefficients of the two independent Casimir functionals $\gamma^{(j)}\in
I(\mathcal{\tilde{G}}^{\ast}),j=1,2,$ gradients:%
\begin{align}
\nabla\gamma^{(1)}(l)  &  \sim(1/z_{x_{1}}+kz_{x_{2}}/z_{x_{1}},-k)^{\intercal
}+O(1/\lambda^{2}),\text{ \ }\label{eq1.36}\\
& \nonumber\\
\nabla\gamma^{(2)}(l)  &  \sim(z_{x_{2}}/z_{x_{1}},-1)^{\intercal}%
+O(1/\lambda^{2}),\nonumber
\end{align}
where $k\neq1$ is a constant.\ Using the Casimir functionals \ (\ref{eq1.36}),
one can construct the simplest two commuting flows
\begin{equation}
\partial\tilde{l}/\partial y=-ad_{\nabla h^{(y)}(\tilde{l})_{+}}^{\ast}%
\tilde{l},\text{\ \ }\partial\tilde{l}/\partial t=-ad_{\nabla h^{(y)}%
(\tilde{l})_{+}}^{\ast}\tilde{l}\ \text{\ \ } \label{eq1.37}%
\end{equation}
with respect to the evolution parameters $y,t\in\mathbb{R},$ where
\begin{align}
\nabla h^{(y)}(l)_{+}  &  :=(\lambda\nabla\gamma^{(1)}(l))_{+}=(\lambda
/z_{x_{1}}+\lambda kz_{x_{2}}/z_{x_{1}},-\lambda k)^{\intercal}:=(\lambda
u_{y},-\lambda k)^{\intercal},\nonumber\\
& \nonumber\\
\text{ \ \ }\nabla h^{(t)}(l)_{+}  &  :=(\lambda\nabla\gamma^{(2)}%
(l))_{+}=(\lambda z_{x_{2}}/z_{x_{1}},-\lambda)^{\intercal}:=(\lambda
u_{t},-\lambda)^{\intercal} \label{eq1.38}%
\end{align}
for some function $u\in C^{\infty}(\mathbb{R}^{2}\times\mathbb{T}%
^{2};\mathbb{R}).$ From relationships \ \ (\ref{eq1.38}), as a result of the
commutativity of the flows \ (\ref{eq1.37}), \ one derives the equivalent Lax
type relationship \ (\ref{eq1.20}) for the vector fields, namely
\begin{equation}
\tilde{A}_{\nabla h_{+}^{(y)}}:=\lambda u_{y}\partial/\partial x_{1}%
-k\lambda\partial/\partial x_{2},\text{ \ \ }\tilde{A}_{\nabla h_{+}^{(t)}%
}:=\lambda u_{t}\partial/\partial x_{1}-\lambda\partial/\partial x_{2},
\label{eq1.39}%
\end{equation}
which can be rewritten as the compatibility condition for the following vector
field equations:
\begin{equation}
\frac{\partial\psi}{\partial t}+\lambda u_{t}\frac{\partial\psi}{\partial
x_{1}}-\lambda\frac{\partial\psi}{\partial x_{2}}=0,\text{ \ \ \ }%
\frac{\partial\psi}{\partial y}+\lambda u_{y}\frac{\partial\psi}{\partial
x_{1}}-k\lambda\frac{\partial\psi}{\partial x_{2}}=0, \label{eq1.40}%
\end{equation}
satisfied for $\psi\in C^{\infty}(\mathbb{R}^{2}\times\mathbb{T}%
^{2};\mathbb{C}),$ any $(t,y;x_{1},x_{2})\in\mathbb{R}^{2}\times\mathbb{T}%
^{2}$ and all $\lambda\in\mathbb{C}.$ The resulting equation is then
\begin{equation}
u_{yx_{2}}-u_{t}u_{yx_{1}}+u_{y}u_{tx_{1}}+ku_{tx_{2}}=0, \label{eq1.41}%
\end{equation}
which reduces at $k=0$ to the Alonso--Shabat heavenly equation \ (\ref{eq1.34}).

\begin{remark}
It is interesting to observe that the seed elements $\tilde{l}\in
\mathcal{\tilde{G}}^{\ast}$ of the examples presented above have the
differential geometric structure:%
\begin{equation}
\tilde{l}=\eta\text{ }d\rho,\ \label{eq1.42}%
\end{equation}
where $\ \eta$ and $\rho\in C^{\infty}(\mathbb{R}^{2}\times(\mathbb{C\times
T}^{2});\mathbb{C})$ are some smooth functions. For instance,
\begin{align*}
\tilde{l}  &  =\ d(\lambda x-2u\ )\text{\ \ \ -\ Mikhalev--Pavlov equation,}\\
& \\
\tilde{l}  &  =d(\lambda x_{1}+\lambda x_{2}-u_{x_{2}}+u_{x_{1}})\text{ \ -
Pleba\'{n}ski heavenly equation,}\\
& \\
\tilde{l}  &  =u_{x_{1}x_{2}}\xi du_{x_{2}},\xi:=\left(  \mu\ [\gamma
(\mu+\beta)]^{-1}+\alpha^{-1}-\mu\lbrack\beta(\mu-\gamma)]^{-1}\right)  \text{
- general heavenly equation,}\\
& \\
\tilde{l}  &  =(\lambda+1)z_{x_{1}}dz\text{ - \ Alonso--Shabat heavenly
equation.}%
\end{align*}

\end{remark}

\section{\label{Sec_4}The generalized heavenly type Lie-algebraic structures}

It is well known that the loop Lie algebra $\ \mathcal{\tilde{G}}:=$
$\widetilde{diff}(\mathbb{T}^{n})$ can be centrally extended as $\widehat
{\mathcal{G}}:=(\widetilde{diff}(\mathbb{T}^{n});\mathbb{R}^{1})$ only
\cite{GeFu} for the case $n=1,$ where for any two elements $(\tilde{a}%
;\alpha)$ and $(\tilde{b};\beta)\in\widehat{\mathcal{G}}$\ \ the commutator
\begin{equation}
\lbrack(\tilde{a};\alpha),(\tilde{b};\beta)]=([\tilde{a},\tilde{b}];\omega
_{2}(\tilde{a},\tilde{b}))\in\mathcal{\tilde{G}} \label{eq2.1}%
\end{equation}
and the $2$-cocycle $\omega_{2}:\mathcal{\tilde{G}}\times\mathcal{\tilde{G}%
}\rightarrow\mathbb{R}^{1}$ satisfies the condition
\begin{equation}
\omega_{2}([\tilde{a},\tilde{b}],\tilde{c})+\omega_{2}([\tilde{b},\tilde
{c}],\tilde{a})+\omega_{2}([\tilde{c},\tilde{a}],\tilde{b})=0 \label{eq2.2}%
\end{equation}
for any $\tilde{a},\tilde{b}$ and $\tilde{c}\in\mathcal{\tilde{G}}.$ For the
case $n=1$, the Gelfand--Fuchs $2$-cocycle \cite{GeFu} \ on the loop Lie
algebra $\mathcal{\tilde{G}}$ \ equals the expression
\begin{equation}
\omega_{2}(\tilde{a},\tilde{b})=\frac{1}{2\pi i}\oint\limits_{\mathbb{S}^{1}%
}\left(  \frac{\partial^{2}a(x;\lambda)}{\partial x^{2}},\frac{\partial
b(x;\lambda)}{\partial x}\right)  _{H^{0}}\frac{d\lambda}{\lambda^{p}}
\label{eq2.3}%
\end{equation}
for any vector fields $\tilde{a}=a(x;\lambda)\frac{\partial}{\partial
x},\tilde{b}=b(x;\lambda)\frac{\partial}{\partial x}\in\mathcal{\tilde{G}}$
\ on $\mathbb{T}^{1}$ and a fixed integer $p\in\mathbb{Z}.$

The integrable dynamical systems related to this central extension were
described in detail in \cite{Misi}. Concerning a further generalization of the
multi-dimensional case related to the loop group $\mathcal{\tilde{G}}$ for
$n\in\mathbb{Z}_{+}$ one can proceed in the following natural way: as the Lie
algebra $\mathcal{\tilde{G}=}$ $\widetilde{diff}(\mathbb{T}^{n})$ consists of
the elements formally depending additionally on the \textquotedblleft
spectral\textquotedblright\ variable $\lambda\in\mathbb{C}^{1},$ one can
extend the\ basic Lie structure on $\mathcal{G=}diff(\mathbb{T}^{n})$ to that
on the adjacent holomorphic in $\lambda\in\mathbb{S}_{\pm\text{ }}^{1}$ Lie
algebra $\mathcal{\bar{G}}:=diff_{hol}(\ \mathbb{C\times T}^{n})\subset
diff(\ \mathbb{C\times T}^{n})\ $of vector fields on $\ \mathbb{C\times T}%
^{n}.$ This has elements representable as $\ \bar{a}(x;\lambda):=<a(x;\lambda
),\frac{\partial}{\partial\mathrm{x}}>=\sum\limits_{j=1}^{n}a_{j}%
(x;\lambda)\frac{\partial}{\partial x_{j}}+a_{0}(x;\lambda)\frac{\partial
}{\partial\lambda}\in\mathcal{\bar{G}}\ \ $ for some holomorphic in
$\lambda\in\mathbb{S}_{\pm\text{ }}^{1}$ vectors $\ a(x;\lambda)\ \in
\mathbb{E}\times\mathbb{E}^{n}$ for all $\ x\ \in\mathbb{T}^{n},$ where
$\frac{\partial}{\partial\mathrm{x}}:=$ $(\frac{\partial}{\partial\lambda
},\frac{\partial}{\partial x_{1}},\frac{\partial}{\partial x_{2}}%
,...,\frac{\partial}{\partial x_{n}})^{\intercal}$ is$\ $the generalized
Euclidean vector gradient with respect to the vector variable\ $\mathrm{x}%
:=(\lambda,x)\in$ $\mathbb{C\times T}^{n}.$

It is now important to mention that the Lie algebra $\mathcal{\bar{G}}$
$\subset diff(\ \mathbb{C\times T}^{n})$\ also splits into the direct sum of
two subalgebras:
\begin{equation}
\mathcal{\bar{G}}=\mathcal{\bar{G}}_{+}\oplus\mathcal{\bar{G}}_{-},
\label{eq2.7}%
\end{equation}
allowing to introduce on it the classical $\mathcal{R}$-structure:
\begin{equation}
\lbrack\bar{a},\bar{b}]_{\mathcal{R}}:=[\mathcal{R}\bar{a},\bar{b}]+[\bar
{a},\mathcal{R}\bar{b}] \label{eq2.8}%
\end{equation}
for any $\bar{a},\bar{b}\in\mathcal{\bar{G}},$ where
\begin{equation}
\mathcal{R}:=(P_{+}-P_{-})/2, \label{eq2.9}%
\end{equation}
and
\begin{equation}
P_{\pm}\mathcal{\bar{G}}:=\mathcal{\bar{G}}_{\pm}\subset\mathcal{\bar{G}}.
\label{eq2.10}%
\end{equation}
The space $\mathcal{\bar{G}}^{\ast}\simeq\Lambda^{1}(\ \mathbb{C\times T}%
^{n}\mathbb{)},$ adjoint to the Lie algebra $\mathcal{\bar{G}}\ $of vector
fields on $\ \mathbb{C\times T}^{n},$ can be functionally identified with
$\mathcal{\bar{G}}$\ $\ $subject to the Sobolev type metric
\begin{equation}
(\bar{l},\ \bar{a})\ =\frac{1}{2\pi i}\oint\limits_{\mathbb{S}^{1}}%
\lambda^{-p}d\lambda(l,a)_{H^{q}}, \label{eq2.11}%
\end{equation}
where $p\in\mathbb{Z},q\in\mathbb{Z}_{+}$ and for arbitrary $\bar
{l}:=<l(x;\lambda),d\mathrm{x}>=\
{\displaystyle\sum\limits_{j=\overline{0,n}}}
l_{j}(x;\lambda)d\mathrm{x}_{j}\ \in\mathcal{\bar{G}}^{\ast},$ $\bar{a}=%
{\displaystyle\sum\limits_{j=\overline{0,n}}}
a_{j}(x;\lambda),\frac{\partial}{\partial\mathrm{x}_{j}}$ $\in\mathcal{\bar
{G}}\ $ one defines
\begin{equation}
(l,a)_{H^{q}}=\sum\limits_{j=0}^{n}\sum\limits_{|\alpha|=0}^{q}\int
\limits_{\mathbb{T}^{n}}dx\frac{\partial^{|\alpha|}l_{j}}{\partial
\mathrm{x}^{\alpha}}\frac{\partial^{|\alpha|}a_{j}}{\partial\mathrm{x}%
^{\alpha}}. \label{eq2.12}%
\end{equation}
In particular,\ for $q=0$ one has $(l,a)_{H^{0}}=\int\limits_{\mathbb{T}^{n}%
}dx\sum\limits_{j=0}^{n}l_{j}\ a_{j},$ the case which will be mainly chosen.
Then for arbitrary $f,g\in\mathrm{D}(\mathcal{\bar{G}}^{\ast}),$ one can
determine two Lie--Poisson brackets
\begin{equation}
\{f,g\}:=(\bar{l},[\nabla f(\bar{l}),\nabla g(\bar{l})])\ \label{eq2.13}%
\end{equation}
\ and
\begin{equation}
\{f,g\}_{\mathcal{R}}:=(\bar{l},[\nabla f(\bar{l}),\nabla g(\bar
{l})]_{\mathcal{R}})\ , \label{eq2.14}%
\end{equation}
where at any seed element $\bar{l}\in\mathcal{\bar{G}}^{\ast}$ the gradient
element $\nabla f(\bar{l})$ and $\nabla g(\bar{l})\in\mathcal{\bar{G}}$ are
calculated with respect to the metric (\ref{eq2.11}).

Now let us assume that a smooth function $\gamma\in I(\mathcal{\bar{G}}^{\ast
})$ is a Casimir invariant, that is \
\begin{equation}
ad_{\nabla\gamma(\bar{l})}^{\ast}\bar{l}=0 \label{eq2.15}%
\end{equation}
for a chosen seed element $\bar{l}\in\mathcal{\bar{G}}^{\ast}.$ As the adjoint
mapping $ad_{\nabla f(\bar{l})}^{\ast}\bar{l}$ \ for any $f\in\mathrm{D}%
(\mathcal{\bar{G}}^{\ast})$ can be rewritten in the reduced form as
\begin{equation}
ad_{\nabla f(\bar{l})}^{\ast}(\bar{l})=\ \ \left\langle \frac{\partial
}{\partial\mathrm{x}},\nabla f(l)\right\rangle \bar{l}+\sum_{j=1}%
^{n}\left\langle \left\langle l,\ \frac{\partial}{\partial\mathrm{x}}\ \nabla
f(l)\ \right\rangle ,d\mathrm{x}\right\rangle , \label{eq2.16}%
\end{equation}
where $\nabla f(\bar{l}):=<\nabla f(l),\frac{\partial}{\partial\mathrm{x}}%
>.$\ For the Casimir function $\gamma\in\mathrm{D}(\mathcal{\bar{G}}^{\ast})$,
the condition (\ref{eq2.15}) is then equivalent to the equation
\begin{equation}
\ \ l\left\langle \frac{\partial}{\partial\mathrm{x}},\nabla\gamma
(l)\right\rangle \ +\left\langle \nabla\gamma(l),\frac{\partial}%
{\partial\mathrm{x}}\right\rangle l+\ \ \left\langle l,(\frac{\partial
}{\partial\mathrm{x}}\nabla\gamma(l))\right\rangle =0, \label{eq2.18}%
\end{equation}
which should be solved analytically. In the case when an element \ $\bar{l}%
\in\mathcal{\bar{G}}^{\ast}$ is singular as $|\lambda|\rightarrow\infty,$ one
can consider the general asymptotic expansion
\begin{equation}
\nabla\gamma:=\nabla\gamma^{(p)}\ \sim\lambda^{p}\sum\limits_{j\in
\mathbb{Z}_{+}}\nabla\gamma_{j}^{(p)}\lambda^{-j}\ \label{eq2.19}%
\end{equation}
for some suitably chosen $p\in\mathbb{Z}_{+},$ and upon substituting
\ (\ref{eq2.19}) into the equation \ (\ref{eq2.18}), one can solve it recurrently.

Now let $h^{(y)},h^{(t)}\in I(\mathcal{\bar{G}}^{\ast})$ be such Casimir
functions for which \ the Hamiltonian vector field generators
\begin{equation}
\nabla h_{+}^{(y)}(l):=(\ \nabla\gamma^{(p_{y})}(l))_{+},\ \ \ \ \ \nabla
h_{+}^{(t)}(l):=(\ \nabla h^{(p_{t})}(l))_{+} \label{eq2.20}%
\end{equation}
are, respectively, defined for special integers $p_{y},p_{t}\in\mathbb{Z}%
_{+}.$ These invariants generate, owing to the Lie--Poisson bracket
(\ref{eq2.14}), for the case $q=0$ the following commuting flows
\begin{align}
\partial l/\partial t\text{ }  &  =-\ \left\langle \frac{\partial}%
{\partial\mathrm{x}},\nabla h_{+}^{(t)}(l)\right\rangle l-\ \left\langle
l,(\frac{\partial}{\partial\mathrm{x}}\nabla h_{+}^{(t)}(l))\right\rangle
,\nonumber\\
& \label{eq2.21}\\
\partial l/\partial y\text{ }  &  =-\ \left\langle \frac{\partial}%
{\partial\mathrm{x}},\nabla h_{+}^{(y)}(l)\right\rangle l-\ \left\langle
l,(\frac{\partial}{\partial\mathrm{x}}\nabla h_{+}^{(y)}(l))\right\rangle
>,\nonumber
\end{align}
where $y,t\in\mathbb{R}$ are the corresponding evolution parameters. Since the
invariants $\ \ h^{(y)},h^{(t)}\in I(\mathcal{\bar{G}}^{\ast})$ commute with
respect to the Lie--Poisson bracket (\ref{eq2.14}), the flows (\ref{eq2.21})
also commute, implying that the corresponding Hamiltonian vector field
generators
\begin{equation}
\bar{A}_{\nabla h_{+}^{(t)}}:=\left\langle \frac{\partial}{\partial\mathrm{x}%
},\nabla h_{+}^{(t)}(l)\right\rangle ,\ \ \ \ \ \bar{A}_{\nabla h_{+}^{(y)}%
}:=\left\langle \frac{\partial}{\partial\mathrm{x}},\nabla h_{+}%
^{(y)}(l)\right\rangle \label{eq2.22}%
\end{equation}
satisfy the Lax compatibility condition
\begin{equation}
\frac{\partial}{\partial y}\bar{A}_{\nabla h_{+}^{(t)}}-\frac{\partial
}{\partial t}\bar{A}_{\nabla h_{+}^{(y)}}=[\bar{A}_{\nabla h_{+}^{(t)}}%
,\bar{A}_{\nabla h_{+}^{(y)}}] \label{eq2.23}%
\end{equation}
for all $y,t\in\mathbb{R}.$ On the other hand, the condition (\ref{eq2.23}) is
equivalent to the compatibility condition of two linear equations
\begin{equation}
(\frac{\partial}{\partial t}+\bar{A}_{\nabla h_{+}^{(t)}})\psi
=0,\ \ \ \ \ (\frac{\partial}{\partial y}+\bar{A}_{\nabla h_{+}^{(y)}})\psi=0
\label{eq2.24}%
\end{equation}
for a function $\psi\in C^{\infty}(\mathbb{R}\times\mathbb{T}^{n};\mathbb{C})$
for all $y,t\in\mathbb{R}$ and any $\lambda\in\mathbb{C}.$

The above can be formulated as the following key result:

\begin{proposition}
\label{Prop_2.1} Let a seed vector field be $\bar{l}\in\mathcal{\bar{G}}%
^{\ast}$ and $h^{(y)},h^{(t)}\in I(\mathcal{\bar{G}}^{\ast})$ be Casimir
functions subject to the metric $(\cdot,\cdot)$ on the loop Lie algebra
$\mathcal{\bar{G}}$ and the natural coadjoint action on the loop co-algebra
$\mathcal{\bar{G}}^{\ast}.$ Then the following dynamical systems
\begin{equation}
\partial\bar{l}/\partial y=-ad_{\nabla h_{+}^{(y)}(\bar{l})}^{\ast}\bar
{l},\text{ \ \ }\partial\bar{l}/\partial t=-ad_{\nabla h_{+}^{(t)}(\bar{l}%
)}^{\ast}\bar{l} \label{eq2.24a}%
\end{equation}
are commuting Hamiltonian flows for all $y,t\in\mathbb{R}.$ \ Moreover, the
compatibility condition of these flows is equivalent to the vector fields
representation
\begin{equation}
(\partial/\partial t+\bar{A}_{\nabla h_{+}^{(t)}})\psi=0,\ \ \ \ \ (\partial
/\partial y+\bar{A}_{\nabla h_{+}^{(y)}})\psi=0, \label{eq2.24b}%
\end{equation}
where $\psi\in C^{\infty}(\mathbb{R}^{2}\times\mathbb{C\times T}%
^{n};\mathbb{C})$ and the vector fields $\bar{A}_{\nabla h_{+}^{(y)}},\bar
{A}_{\nabla h_{+}^{(t)}}\in\mathcal{\bar{G}}$ \ are given by the expressions
(\ref{eq2.22}) and (\ref{eq2.20}).
\end{proposition}

\begin{remark}
\label{Rem2.2} \textit{\ As mentioned above, the expansion (\ref{eq2.19}) is
effective if a chosen seed element } $\bar{l}\in\mathcal{\bar{G}}^{\ast}%
$\textit{\ is singular as} $|\lambda|\rightarrow\infty.$ \textit{In the case
when it is singular as} $|\lambda|\rightarrow0,$ \textit{the expression
(\ref{eq2.19}) should be replaced by the expansion
\begin{equation}
\mathit{\nabla}\gamma^{(p)}(l)\mathit{\sim\lambda}^{-p}\sum\limits_{j\in
\mathbb{Z}_{+}}\mathit{\nabla}\gamma_{j}^{(p)}(l)\mathit{\lambda}%
^{j}\label{eq2.24c}%
\end{equation}
for suitably chosen integers \ \ }$p\in\mathbb{Z}_{+},$ \ \textit{and the
reduced Casimir function gradients then are given by the Hamiltonian vector
field generators
\begin{equation}
\nabla h_{-}^{(y)}(l):=\lambda(\lambda^{-p_{y}-1}\nabla\gamma^{(p_{y}%
)}(l))_{-},\ \ \ \ \ \nabla h_{-}^{(t)}(l):=\lambda(\lambda^{-p_{t}-1}%
\nabla\gamma^{(p_{t})}(l))_{-}\label{eq2.24d}%
\end{equation}
for suitably chosen positive integers }$p_{y},p_{t}\in\mathbb{Z}_{+}$
\textit{and the corresponding Hamiltonian flows are, respectively, written as
\begin{equation}
\partial\bar{l}/\partial t=ad_{\triangledown h_{-}^{(t)}(\bar{l})}^{\ast}%
\bar{l},\ \ \partial\bar{l}/\partial y=ad_{\triangledown h_{-}^{(y)}(\bar{l}%
)}^{\ast}\bar{l}.\label{eq2.24e}%
\end{equation}
}
\end{remark}

As in Section \ref{Sec_2} the Proposition \ref{Prop_2.1} abov e makes it
possible to describe the B\"{a}cklund transformations \ between two special
solution sets for the\ dispersionless heavenly equations resulting from the
\ Lax compatibility condition \ (\ref{eq2.24a}). Let \ a diffeomorphism
$\xi\in Diff(\mathbb{C\times T}^{n})$ \ be such that a seed loop differential
form $\bar{l}(\lambda,x)\in\mathcal{\bar{G}}^{\ast}\simeq\Lambda
^{1}(\ \mathbb{C\times T}^{n})$ satisfies the invariance condition
\begin{equation}
\ \bar{l}(\xi(\mathrm{x};\mu))=k\bar{l}(\mathrm{\bar{x}}) \label{eq2.25}%
\end{equation}
for\ some\ non-zero constant $k\in\mathbb{C}\backslash\{0\},$ any
$\mathrm{x}=(\lambda,x)$ and $\mathrm{\bar{x}}=(\mu,x)\in\mathbb{C\times
T}^{n}\ $and arbitrarily an chosen \ parameter $\ \mu\in\mathbb{C}.$ As the
seed element $\bar{l}(\xi(\mathrm{x};\mu))\ \in\Lambda^{1}(\ \mathbb{C\times
T}^{n})$ satisfies simultaneously the system of compatible equations
\ (\ref{eq2.24a}), the loop diffeomorphism $\xi\in Diff(\ \mathbb{C\times
T}^{n}),$ found analytically from the invariance condition \ (\ref{eq2.25}),
satisfies the compatible system of vector field equations \
\[
\frac{\partial}{\partial t}\xi=\nabla h_{+}^{(t)}(l),\text{ \ }\frac{\partial
}{\partial y}\xi=\nabla h_{+}^{(y)}(l),
\]
giving rise to the B\"{a}cklund type relationships for the coefficients of the
seed\ loop differential form $\bar{l}\in\mathcal{\bar{G}}^{\ast}$
$\simeq\Lambda_{hol}^{1}(\ \mathbb{C\times T}^{n}).$

It worth mentionoing that, following Ovsienko's scheme \cite{Ovsi-1,Ovsi-2},
one can consider a wider class of integrable heavenly equations, realized as
compatible Hamiltonian flows on the semidirect product of the \ holomorphic
loop Lie algebra $\mathcal{\tilde{G}}$ \ of vector fields on the torus
$\mathbb{T}^{n}$ and its regular co-adjoint space $\mathcal{\tilde{G}}^{\ast
},$ supplemented with naturally related cocycles. We plan to analyze this
aspect of the construction, devised in the present work, in a paper now in
preparation \ (\ref{eq1.20}).

\subsection{Example: Einstein--Weyl metric equation}

Define $\mathcal{\bar{G}}^{\ast}=diff_{hol}(\mathbb{T}^{1}\times\mathbb{C})$
and take the seed element
\[
\tilde{l}=\left(  u_{x}\lambda-2u_{x}v_{x}-u_{y}\right)  dx+\left(
\lambda^{2}-v_{x}\lambda+v_{y}+{v_{x}^{2}}\right)  d\lambda,
\]
which generates with respect to the metric (\ref{eq2.11}) (as before for
$q=0)$ the gradient of the Casimir invariants $h^{(p_{t})},h^{(p_{y})}\in
I(\mathcal{\bar{G}}^{\ast})$ in the form
\begin{align}
\nabla h^{(p_{t})}(l)  &  \sim\ \lambda^{2}(0,1)^{\intercal}+(-u_{x}%
,v_{x})^{\intercal}\lambda\ +(u_{y},u-v_{y})^{\intercal}+O(\lambda
^{-1}),\label{eq2.26}\\
& \nonumber\\
\nabla h^{(p_{y})}(l)  &  \sim\lambda(0,1)^{\intercal}+(-u_{x},v_{x}%
)^{\intercal}\ +(u_{y},-v_{y})^{\intercal}\lambda^{-1}+O(\lambda
^{-2})\nonumber
\end{align}
as $|\lambda|\rightarrow\infty\ $at $p_{t}=2,$ $p_{y}=1.$ \ For the gradients
of the Casimir functions $h^{(t)},h^{(y)}\in I(\mathcal{G}^{\ast}),$
determined by (\ref{eq2.20}) one can easily obtain the corresponding
Hamiltonian vector field generators
\begin{align}
&  \tilde{A}_{\nabla h_{+}^{(t)}}=\left\langle \nabla h_{+}^{(t)}%
(l),\frac{\partial}{\partial\mathrm{x}}\right\rangle =(\lambda^{2}+\lambda
v_{x}+u-v_{y})\frac{\partial}{\partial x}+(-\lambda u_{x}+u_{y})\frac
{\partial}{\partial\lambda},\nonumber\\
& \nonumber\\
&  \tilde{A}_{\nabla h_{+}^{(y)}}=\left\langle \nabla h_{+}^{(y)}%
(l),\frac{\partial}{\partial\mathrm{x}}\right\rangle =(\lambda+v_{x}%
)\frac{\partial}{\partial x}-u_{x}\frac{\partial}{\partial\lambda},
\label{eq2.27}%
\end{align}
satisfying the compatibility condition (\ref{eq2.23}), which is equivalent to
the set of equations
\begin{equation}
\left\{
\begin{array}
[c]{l}%
u_{xt}+u_{yy}+(uu_{x})_{x}+v_{x}u_{xy}-v_{y}u_{xx}=0,\\
\\
v_{xt}+v_{yy}+uv_{xx}+v_{x}v_{xy}-v_{y}v_{xx}=0,
\end{array}
\right.  \label{eq2.28}%
\end{equation}
describing general integrable Einstein--Weyl metric equations \cite{DuMaTo}.

As is well known \cite{MaSa}, the invariant reduction of (\ref{eq2.28}) at
$v=0$ gives rise to the famous dispersionless Kadomtsev--Petviashvili
equation
\begin{equation}
(u_{t}+uu_{x})_{x}+u_{yy}=0, \label{eq2.29}%
\end{equation}
for which the reduced vector field representation (\ref{eq2.24}) follows from
(\ref{eq2.27}) and is given by the vector fields
\begin{align}
&  \bar{A}_{\nabla h_{+}^{(t)}}=(\lambda^{2}+u)\frac{\partial}{\partial
x}+(-\lambda u_{x}+u_{y})\frac{\partial}{\partial\lambda},\label{eq2.30}\\
& \nonumber\\
&  \bar{A}_{\nabla h_{+}^{(y)}}=\lambda\frac{\partial}{\partial x}-u_{x}%
\frac{\partial}{\partial\lambda},\nonumber
\end{align}
satisfying the compatibility condition (\ref{eq2.23}), equivalent to the
equation (\ref{eq2.29}). In particular, one derives from (\ref{eq2.24}) and
\ (\ref{eq2.30}) the vector field compatibility relationships
\begin{equation}%
\begin{array}
[c]{c}%
\frac{\partial\psi}{\partial t}+(\lambda^{2}+u)\frac{\partial\psi}{\partial
x}+(-\lambda u_{x}+u_{y})\frac{\partial\psi}{\partial\lambda}=0\\
\\
\ \frac{\partial\psi}{\partial y}+\lambda\frac{\partial\psi}{\partial x}%
-u_{x}\frac{\partial\psi}{\partial\lambda}=0,
\end{array}
\label{eq2.30a}%
\end{equation}
satisfied for $\psi\in C^{\infty}(\mathbb{R}^{2}\times\mathbb{T}^{1}%
\times\mathbb{C};\mathbb{C})\ $and any $\ y,t\in\mathbb{R},(x,\mathbb{\lambda
)}\in\mathbb{T}^{1}\times\mathbb{C}.$

\subsection{The modified Einstein--Weyl metric equation}

This equation system is
\begin{align}
u_{xt}  &  =u_{yy}+u_{x}u_{y}+u_{x}^{2}w_{x}+uu_{xy}+u_{xy}w_{x}%
+u_{xx}a,\label{eq2.31aa}\\
& \nonumber\\
w_{xt}  &  =uw_{xy}+u_{y}w_{x}+w_{x}w_{xy}+aw_{xx}-a_{y},\nonumber
\end{align}
where $a_{x}:=u_{x}w_{x}-w_{xy},$ and was recently derived in \cite{Szab}, In
this case we take also $\mathcal{\bar{G}}^{\ast}=diff_{hol}(\mathbb{T}%
^{1}\times\mathbb{C}),$ \ yet \ for a seed element $\ \ \tilde{l}%
\in\mathcal{\bar{G}}$ \ \ we choose the form
\begin{align}
&  \tilde{l}=[{{\lambda}^{2}}u_{x}+\left(  2u_{x}w_{x}+u_{y}+3{u}u_{x}\right)
\lambda+2u_{x}\partial_{x}^{-1}{\left.  u_{x}w_{x}\right.  }+2u_{x}%
\partial_{x}^{-1}{\left.  u_{y}\right.  +}\label{eq2.31a}\\
& \nonumber\\
&  +3u_{x}{{w_{x}}^{2}}+2u_{y}w_{x}+6{u}u_{x}w_{x}+2{u}u_{y}+3{{u}^{2}}%
u_{x}-2{a}u_{x}]dx+\nonumber\\
& \nonumber\\
+  &  [{{\lambda}^{2}}+\left(  w_{x}+3{u}\right)  \lambda+2\partial_{x}%
^{-1}{\left.  u_{x}w_{x}\right.  }+2\partial_{x}^{-1}{\left.  u_{y}\right.
}+{{w_{x}}^{2}}+3{u}w_{x}+3{{u}^{2}}-{a}]d\lambda,\nonumber
\end{align}
which with respect to the metric (\ref{eq2.11}) (as before for $q=0)$
generates two Casimir invariants $\gamma^{(j)}\in I(\mathcal{\bar{G}}^{\ast
}),j=\overline{1,2},$ whose gradients are%
\begin{align}
\nabla\gamma^{(2)}(l)  &  \sim\lambda^{2}[(u_{x},-1)^{\intercal}+(uu_{x}%
+u_{y},-u+w_{x})^{\intercal}\lambda^{-1}\ +\label{eq2.31ba}\\
& \nonumber\\
\  &  +(0,uw_{x}-a)^{\intercal}\lambda^{-2}]+O(\lambda^{-1})\ ,\nonumber\\
& \nonumber\\
\nabla\gamma^{(1)}(l)  &  \sim\lambda\lbrack(u_{x},-1)^{\intercal}%
+(0,w_{x})^{\intercal}\ \lambda^{-1}]+O(\lambda^{-1}),\nonumber
\end{align}
as $|\lambda|\rightarrow\infty\ $at $p_{y}=1,p_{t}=2.$ The corresponding
gradients of the Casimir functions $h^{(t)},h^{(y)}\in I(\mathcal{G}^{\ast}),$
determined by (\ref{eq2.20}), generate \ the Hamiltonian vector field
expressions
\begin{align}
&  \nabla h_{+}^{(y)}:=\nabla\gamma^{(1)}(l)|_{+}=(u_{x}\lambda,-\lambda
+w_{x})^{\intercal},\label{eq2.31c}\\
& \nonumber\\
&  \nabla h_{2,+}^{(t)}=\nabla\gamma^{(2)}(l)|_{+}=(u_{x}\lambda^{2}%
+(uu_{x}+u_{y})\lambda,-\lambda^{2}+(w_{x}-u)\lambda+uw_{x}-a)^{\intercal
}.\nonumber
\end{align}
Now one easily obtains from \ (\ref{eq2.31c}) the compatible Lax system of
linear equations%

\begin{align}
&  \frac{\partial\psi}{\partial y}+(-\lambda+w_{x})\frac{\partial\psi
}{\partial x}+u_{x}\lambda\frac{\partial\psi}{\partial\lambda}%
=0,\label{eq2.31d}\\
& \nonumber\\
&  \frac{\partial\psi}{\partial t}+\ (-\lambda^{2}+(\ w_{x}-u)\lambda
+uw_{x}-a)\frac{\partial\psi}{\partial x}+(u_{x}\lambda^{2}+(uu_{x}%
+u_{y})\lambda)\frac{\partial\psi}{\partial\lambda}=0,\nonumber
\end{align}
satisfied for $\psi\in C^{\infty}(\mathbb{R}^{2}\times\mathbb{T}^{1}%
\times\mathbb{C};\mathbb{C})\ $and any $\ y,t\in\mathbb{R},$
$\ (x,\mathbb{\lambda)}\in\mathbb{T}^{1}\times\mathbb{C}.$

\subsection{Example: The Dunajski heavenly equations}

This equation, \ suggested in \cite{Duna}, generalizes the corresponding
anti-self-dual vacuum Einstein equation, which is related to the Pleba\'{n}ski
metric and the celebrated Pleba\'{n}ski \cite{Pleb} second heavenly equation
\ (\ref{eq1.29}). To study the integrability of the Dunajski equations
\begin{align}
u_{x_{1}t}+u_{yx_{2}}+u_{x_{1}x_{1}}u_{x_{2}x_{2}}-u_{x_{1}x_{2}}^{2}-v  &
=0,\label{eq2.44}\\
& \nonumber\\
v_{x_{1}t}+v_{x_{2}y}+u_{x_{1}x_{1}}v_{x_{2}x_{2}}-2u_{x_{1}x_{2}}%
v_{x_{1}x_{2}}  &  =0,\nonumber
\end{align}
where $(u,v)\in C^{\infty}(\mathbb{R}^{2}\times\mathbb{T}^{2};\mathbb{R}%
^{2}),$ $(y,t;x_{1},x_{2})\in\mathbb{R}^{2}\times\mathbb{T}^{2},$ we define
$\mathcal{\bar{G}}^{\ast}:=diff_{hol}^{\ast}(\ \mathbb{C\times T}^{n})$ and
take the following as a seed element $\bar{l}\in\mathcal{\bar{G}}^{\ast}$
\begin{equation}
\tilde{l}=(\lambda+v_{x_{1}}-u_{x_{1}x_{1}}+u_{x_{1}x_{2}})dx_{1}%
+(\lambda+v_{x_{2}}+u_{x_{2}x_{2}}-u_{x_{1}x_{2}})dx_{2}+(\lambda-x_{1}%
-x_{2})d\lambda. \label{eq2.45}%
\end{equation}
With respect to the metric (\ref{eq2.11}) (as before for $q=0),$ the gradients
of two functionally independent Casimir invariants $h^{(p_{y})},h^{(p_{y}%
)\ }\in I(\mathcal{\bar{G}}^{\ast})\ $can be obtained as $|\lambda
|\rightarrow\infty\ $\ in the asymptotic form as
\begin{align}
\nabla h^{(p_{y})\ }(l)  &  \sim\lambda(0,1,0)^{\intercal}+(-v_{x_{1}%
},-u_{x_{1}x_{2}},u_{x_{1}x_{1}})^{\intercal}\ +O(\lambda^{-1}),\label{eq2.46}%
\\
& \nonumber\\
\nabla h^{(p_{t})\ }(l)  &  \sim\lambda(0,0,-1)^{\intercal}+(v_{x_{2}%
},u_{x_{2}x_{2}},-u_{x_{1}x_{2}})^{\intercal}\ +O(\lambda^{-1})\ \nonumber
\end{align}
at $p_{t}=1=p_{y}.$ \ Upon calculating the Hamiltonian vector field
generators
\begin{align}
\nabla h_{+}^{(y)}  &  :=\nabla h^{(p_{y})\ }(l)|_{+}=(-v_{x_{1}}%
,\lambda-u_{x_{1}x_{2}},u_{x_{1}x_{1}})^{\intercal},\label{eq2.47}\\
& \nonumber\\
\nabla h_{+}^{(t)}  &  :=\nabla h^{(p_{t})\ }(l)|_{+}=(v_{x_{2}},u_{x_{2}%
x_{2}},-\lambda-u_{x_{1}x_{2}})^{\intercal},\nonumber
\end{align}
following from the Casimir functions gradients \ (\ref{eq2.46}), \ one easily
obtains the following vector fields
\begin{align}
&  \bar{A}_{\nabla h_{+}^{(t)}}=<\nabla h_{+}^{(t)},\frac{\partial}%
{\partial\mathrm{x}}>=u_{x_{2}x_{2}}\frac{\partial}{\partial x_{1}}%
-(\lambda+u_{x_{1}x_{2}})\frac{\partial}{\partial x_{2}}+v_{x_{2}}%
\frac{\partial}{\partial\lambda},\label{eq2.48}\\
& \nonumber\\
&  \bar{A}_{\nabla h_{+}^{(y)}}=<\nabla h_{+}^{(y)},\frac{\partial}%
{\partial\mathrm{x}}>=(\lambda-u_{x_{1}x_{2}})\frac{\partial}{\partial x_{1}%
}+u_{x_{1}x_{1}}\frac{\partial}{\partial x_{2}}-v_{x_{1}}\frac{\partial
}{\partial\lambda},\nonumber
\end{align}
satisfying the Lax compatibility condition \ (\ref{eq2.23}), which is
equivalent to the the Dunajski \cite{Duna} vector field compatibility
relationships (\ref{eq1.24}) \
\begin{align}
\frac{\partial\psi}{\partial t}+u_{x_{2}x_{2}}\frac{\partial\psi}{\partial
x_{1}}-(\lambda+u_{x_{1}x_{2})}\frac{\partial\psi}{\partial x_{2}}+v_{x_{2}%
}\frac{\partial\psi}{\partial\lambda}  &  =0,\nonumber\\
& \label{eq2.49}\\
\frac{\partial\psi}{\partial y}+(\lambda-u_{x_{1}x_{2}})\frac{\partial\psi
}{\partial x_{1}}+u_{x_{1}x_{1}}\frac{\partial\psi}{\partial x_{2}}-v_{x_{1}%
}\frac{\partial\psi}{\partial\lambda}  &  =0,\nonumber
\end{align}
satisfied for $\psi\in C^{\infty}(\mathbb{R}^{2}\times\mathbb{C\times
}\mathbb{T}^{2};\mathbb{C}),$ any $\ (y,t;x_{1},x_{2})\in\mathbb{R}^{2}%
\times\mathbb{T}^{2}$ and all $\lambda\in\mathbb{C}.$ \ As was mentioned in
\cite{BoDrMa}, the Dunajski equations \ (\ref{eq2.44}) generalize both the
dispersionless Kadomtsev--Petviashvili and Pleba\'{n}ski second heavenly
equations, and is also a Lax integrable Hamiltonian system.

\section{\label{Sec_5}Integrability, bi-Hamiltonian structures and the
classical Lagrange-d'Alembert principle}

It is evident that all evolution flows like (\ref{eq1.18}) or (\ref{eq2.21})
are Hamiltonian with respect to the second Lie--Poisson bracket (\ref{eq2.14})
on the adjoint loop space $\mathcal{\tilde{G}}^{\ast}=\widetilde{diff}^{\ast
}(\mathbb{T}^{n})$ or on the holomorphic subspace $\mathcal{\bar{G}}^{\ast
}=diff_{hol}^{\ast}(\ \mathbb{C\times T}^{n}),$ respectively. Moreover, they
are poly-Hamiltonian on the corresponding functional manifolds, as the related
\ bilinear forms \ (\ref{eq1.1}) and \ (\ref{eq2.11}) are \ marked by integers
$p\in\mathbb{Z}$. This leads to \cite{FaTa} \ an infinite hierarchy of
compatible Poisson structures on the phase spaces, isomorphic, respectively,
to the orbits of a chosen seed element $\tilde{l}\in\mathcal{\tilde{G}}^{\ast
}$ or of a seed element $\bar{l}\in\mathcal{\bar{G}}^{\ast}.$ Taking also into
account that all these Hamiltonian flows possess an infinite hierarchy of
commuting nontrivial conservation laws, one can prove their formal complete
integrability under some naturally formulated constraints. The corresponding
analytical expressions for the infinite hierarchy of conservation laws can be
retrieved from the asymptotic expansion \ (\ref{eq1.22}) for Casimir
functional gradients by employing the well-known \cite{BlPrSa,Olve,FaTa}
formal homotopy technique.

As an arbitrary heavenly type equation is a Hamiltonian system with respect to
both evolution parameters $t,y\in\mathbb{R}^{2},$ one can construct
\cite{Olve,PrMy,BlPrSa} its suitable Lagrangian representation under some
natural constraints. Thus, it is possible to retrieve the corresponding
Poisson structures related to both these evolution parameters $t,y\in
\mathbb{R}^{2},$ which, as follows from the Lie-algebraic analysis in Section
\ref{Sec_2}, are compatible with each other. In this way, one can show that
any heavenly\ type equation is a bi-Hamiltonian integrable system on the
corresponding functional manifold. It should be mentioned here that this
property was introduced by Sergyeyev in (\textit{arXiv:1501.01955}), published
in \cite{Serg}, and rediscovered and applied in detail in \cite{Shef} for
investigating the integrability properties of the general heavenly equation
(\ref{eq2.31}), first suggested by Schief in \cite{Schi} and later studied by
Doubrov and Ferapontov in \cite{DoFe}.

Using our approach in the case of the basic\ loop Lie algebra $\mathcal{\tilde
{G}}=\widetilde{diff}(\mathbb{T}^{n})$ one needs to recall that a seed
element\ $\tilde{l}(\lambda)\in\mathcal{\tilde{G}}^{\ast},\lambda\in
\mathbb{C},$ generates the commuting Hamiltonian flows
\begin{equation}
\partial\ \tilde{l}(\lambda)/\partial t_{j}:=-ad_{\nabla h_{+}^{(j)}%
(\ \tilde{l})}^{\ast}\ \tilde{l}(\lambda) \label{eq3.2}%
\end{equation}
for any $j\in\mathbb{Z}_{+}.$ \ Taking into account that the element
$\tilde{l}(\lambda)\in\mathcal{\tilde{G}}^{\ast},$ \ the hierarchy of flows
(\ref{eq3.2}) can be equivalently rewritten\ differential geometrically as the
generating vector field%
\begin{equation}%
\begin{array}
[c]{c}%
\frac{\partial}{\partial t}\tilde{l}(\lambda)(\tilde{Y})=-\frac{\mu}%
{\mu-\lambda}\tilde{l}(\lambda)([\nabla h(\ \tilde{l}(\mu)),\tilde{Y}])=\\
\\
=\ -i_{\frac{\mu}{\mu-\lambda}\nabla h(\ \tilde{l}(\mu))}d\tilde{l}%
(\lambda)(\tilde{Y})-d\ (\tilde{l}(\lambda)(\frac{\mu}{\mu-\lambda}\nabla
h(\tilde{l}(\mu)))(\tilde{Y})-\ \\
\\
-\frac{\mu}{\mu-\lambda}\ \left\langle \ d/dx,\nabla h(\ \tilde{l}%
(\mu))\right\rangle \ \ \tilde{l}(\lambda)(\tilde{Y}))
\end{array}
\label{eq3.3}%
\end{equation}
for any $\ \tilde{Y}\in\mathcal{\tilde{G}},$ where $\frac{\partial}{\partial
t}:=\sum_{j\in\mathbb{Z}_{+}}\mu^{-j}\frac{\partial}{\partial t_{j}}\ $\ and
$\ \mu\in\mathbb{C}\ $ is such that $\ |\lambda/\mu|<1\ $as $\lambda
,\mu\rightarrow\infty.$Recall now that for the space$\ \ \widetilde
{diff}(\mathbb{T}^{n})\simeq\tilde{\Gamma}(T(\mathbb{T}^{n}))$ \ and space
\ $\ \widetilde{diff}(\mathbb{T}^{n})^{\ast}\simeq\tilde{\Lambda}%
^{1}(\mathbb{T}^{n}),$ \ the generating relationship \ (\ref{eq3.3}) can be
easily rewritten as the following evolution equation
\begin{equation}
\frac{d}{dt}\tilde{l}(\lambda)=-\frac{\mu}{\mu-\lambda}\ \left\langle
(\partial/\partial x,\nabla h(\ l(\mu)))\right\rangle \ \ \tilde{l}%
(\lambda):=-\ \tilde{l}(\lambda)\ \operatorname{div}\tilde{K}(\mu)
\label{eq3.4}%
\end{equation}
\ on the seed element $\tilde{l}(\lambda)\in\tilde{\Lambda}^{1}(\mathbb{T}%
^{n}).$ Here $d/dt:=\partial/\partial t+L_{\tilde{K}(\mu)}$ $\ $and
$\ L_{\tilde{K}(\mu)}=i_{\tilde{K}(\mu)}d+di_{\tilde{K}(\mu)}$ denotes here
the well-known \cite{AbMa,BlPrSa,Godb} Cartan expression for the derivation
along the vector field
\begin{equation}
\tilde{K}(\mu):=\frac{\mu}{\mu-\lambda}\nabla h(\ \tilde{l}(\mu))=\frac{\mu
}{\mu-\lambda}<\nabla h(\ l(\mu)),\frac{\partial}{\partial x}>, \label{eq3.5}%
\end{equation}
which holds asymptotically as $\mu,\lambda\rightarrow\infty,|\lambda/\mu|<1,$
and is equivalent to the so called Lax--Sato hierarchy of equations, studied
in \cite{TaTa-1,TaTa-2, Bogd,BoKo,BoPa} for the generating vector fields
function \ (\ref{eq3.4}).\ We plan to study this and other related algebraic
aspects of these equations in more detail in a work under preparation.

The expression \ (\ref{eq3.4}) allows the following interesting mechanical
Lagrange--d'Alembert type principle \cite{AbMa} interpretation. Namely, the
seed differential form $\tilde{l}(\lambda)=<l(x;\lambda),\delta x>$
$\ \in\tilde{\Lambda}^{1}(\mathbb{T}^{n}),$ can be considered as a virtual
infinitesimal work, performed by the \textquotedblleft\textit{virtual
force}\textquotedblright\ $l(x;\lambda)\in\tilde{T}^{\ast}(\mathbb{T}^{n})$ at
point $x\in\mathbb{T}^{n}\ $\ for any $\lambda\in\mathbb{C}$ \ along the
infinitesimal path $\delta x\in\mathbb{T}^{n}.$ The whole infinitesimal
\textit{\textquotedblleft virtual work\textquotedblright} $\ \delta W(t),$
performed by this force within a moving 1-connected arbitrary open domain
$\Omega_{t}\subset\mathbb{T}^{n}$ with smooth boundary $\partial\Omega_{t},$
$t\in\mathbb{R},$ equals \
\begin{equation}
\delta W(t):=\int_{\Omega_{t}}<l(x(t);\lambda),\delta x(t)>d^{n}x(t),
\label{eq3.6}%
\end{equation}
where the evolution of points $x(t)\in\Omega_{t}$ is naturally determined by
the vector field (\ref{eq3.5})
\begin{equation}
\frac{dx(t)\ }{dt}=\frac{\mu}{\mu-\lambda}\nabla h(\ l(\mu))(t;x(t))
\label{eq3.7}%
\end{equation}
and the Cauchy data
\[
x(t)|_{t=0}=x_{0}\in\Omega_{0}%
\]
for an arbitrarily chosen open 1-connected domain $\Omega_{0}\subset
\mathbb{T}^{n}$ with the smooth boundary $\partial\Omega_{0}.$ Then the
Lagrange--d'Alembert principle in mechanics says that the infinitesimal
virtual work \ (\ref{eq3.7}) equals zero for all moments of time, \ that is
$\delta W(t)=0=$ $\delta W(0)$\ for all $t\in\mathbb{R}.\ $ To check that it
is really true, let us calculate the temporal derivative of the expression
\ (\ref{eq3.6}):%
\begin{equation}%
\begin{array}
[c]{c}%
\frac{d}{dt}\delta W(t)=\frac{d}{dt}\int_{\Omega_{t}}<l(x(t);\lambda),\delta
x(t)>d^{n}x(t)=\\
\\
=\frac{d}{dt}\int_{\Omega_{0}}<l(x(t);\lambda),\delta x(t)>|\frac
{\partial(x(t)}{\partial x_{0}}|d^{n}x_{0}=\int_{\Omega_{0}}\frac{d}%
{dt}(<l(x(t);\lambda),\delta x(t)>|\frac{\partial(x(t)}{\partial x_{0}}%
|)d^{n}x_{0}=\\
\\
=\int_{\Omega_{0}}[(\partial/\partial t+L_{\tilde{K}(\mu)})<l(x(t);\lambda
),\delta x(t)>+<l(x(t);\lambda),\delta x(t)>\operatorname{div}\tilde{K}%
(\mu)]|\frac{\partial(x(t)}{\partial x_{0}}|d^{n}x_{0}=\\
\\
=\int_{\Omega_{t}}[(\partial/\partial t+L_{\tilde{K}(\mu)})<l(x(t);\lambda
),\delta x(t)>+<l(x(t);\lambda),\delta x(t)>\operatorname{div}\tilde{K}%
(\mu)]d^{n}x(t)=\\
\\
=\int_{\Omega_{t}}[\left(  \partial/\partial t+L_{\tilde{K}(\mu)}%
+\operatorname{div}\tilde{K}(\mu)\right)  <l(x(t);\lambda),\delta
x(t)>]d^{n}x(t)=0,
\end{array}
\label{eq3.8}%
\end{equation}
owing to the equation \ (\ref{eq3.4}). Thus, if at $t=0$ one has $\delta
W(0)=0,$ the infinitesimal work $\delta W(t)=0$ for all $t\in\mathbb{R},$
proving the Lagrange--d'Alembert principle for the generating evolution
equation (\ref{eq3.4}).

\bigskip

\section{Acknowledgements}

The authors cordially thank Prof. M. B\l aszak, Prof. J. Cie\'{s}linski and
Prof. A. Sym for their cooperation and useful discussions of the results in
this paper during the Workshop \textquotedblleft Nonlinearity and
Geometry\textquotedblright\ held 20-23 January 2017 in Warsaw. A.P. is
especially indebted to Prof. M. Pavlov, Prof. W.K. Schief and Ya. G. Prytula
for mentioning important references, which were very helpful when preparing
the manuscript. He is also greatly indebted to Prof. V.E. Zakharov (University
of Arizona, Tucson) and Prof. J. Szmigelski (University of Saskatchewan,
Saskatoon) for their interest in the work and instructive discussions during
the XXXV Workshop on Geometric Methods in Physics, held 26.06-2.07.2016 in
Bia\l owie\.{z}a, Poland.

\end{document}